\documentclass[draft,12pt,reqno,a4paper]{amsart}

\usepackage[margin=2.5cm,footskip=1cm]{geometry}

\usepackage{amsmath}
\usepackage{amsfonts}
\usepackage{amsthm}
\usepackage{mathtools}

\usepackage{enumerate}
\usepackage[latin1]{inputenc}
\usepackage[shortlabels]{enumitem}

\usepackage{graphicx}
\usepackage{verbatim}

\DeclareMathOperator*{\wslim}{w^\star--lim}

\DeclareMathOperator*{\vGlim}{{\mathcal G}--lim}
\DeclareMathOperator*{\wvGlim}{{w--\mathcal G}--lim}
\DeclareMathOperator*{\wBstarlim}{{w^\star--B^*}--lim}
\DeclareMathOperator*{\wLstarlim}{{w^\star}-- {\mathit L^2_{\mathit m}}--{lim}}

     \newcommand{\supp}{\operatorname{supp}}

     \newcommand{\tr}{\operatorname{tr}}

     \newcommand{\Ran}{\operatorname{Ran}}

\newcommand{\Opw}{\operatorname{Op\!^w}}
     \newcommand{\N}{{\mathbb{N}}}

     \newcommand{\R}{{\mathbb{R}}}
     
     \newcommand{\C}{{\mathbb{C}}}

\newcommand{\e}{{\rm e}}

\renewcommand{\i}{\mathrm{i}}
\renewcommand{\c}{\mathrm{c}}
\newcommand{\clos}{{\rm clos}}
\renewcommand{\d}{{\rm d}}
\newcommand{\D}{{\rm D}}

\DeclareMathOperator\unif{unif}

\newcommand{\rad}{{\rm rad}}

\renewcommand{\tr}{{\rm tr}}

\renewcommand{\Re}{{\rm Re}\,}
\renewcommand{\Im}{{\rm Im}\,}

\DeclarePairedDelimiter\inp\langle\rangle

\newcommand\parn[2][]{#1 ( #2#1)}
\newcommand\parb[2][]{#1 \big ( #2#1\big )}
\newcommand\parbb[2][]{#1 \Big ( #2#1\Big )}

\renewcommand{\exp}{{\rm exp}}
\newcommand{\mand}{\text{ and }}
\newcommand{\mfor}{\text{ for }}
\newcommand{\mforall}{\text{ for all }}

\newcommand{\vB}{{\mathcal B}}

\newcommand{\vE}{{\mathcal E}}

\newcommand{\vG}{{\mathcal G}}

\newcommand{\vH}{{\mathcal H}}

\newcommand{\vV}{{\mathcal V}}
\newcommand{\vW}{{\mathcal W}}

\newcommand{\vS}{{\mathcal S}}

     \theoremstyle{plain}
     \newtheorem{thm}{Theorem}[section]
     \newtheorem{proposition}[thm]{Proposition}
     \newtheorem{lemma}[thm]{Lemma}
      \newtheorem{corollary}[thm]{Corollary}
     \theoremstyle{definition}
     
     \newtheorem{defn}[thm]{Definition}

     \newtheorem{cond}[thm]{Condition}

     \newtheorem{remark}[thm]{Remark}

\newtheorem*{remarks*}{Remarks}
\newtheorem*{remark*}{Remark}

     \numberwithin{equation}{section}

\title[Renormalized two-body low-energy
  scattering]{Renormalized two-body low-energy
  scattering}

\author{E. Skibsted}
\address[E. Skibsted]{Department of Mathematics \\
Aarhus Universitet\\ Ny Munkegade  8000 Aarhus C,
Denmark}
\email{skibsted@imf.au.dk}

\begin{document}

\begin{abstract}
For a class of
long-range potentials, including ultra-strong perturbations of the attractive
Coulomb potential in
dimension $d\geq3$, we introduce a stationary scattering theory for
Schr\"odinger operators  which
is regular at zero energy. In particular it is well defined at this
energy, and we use it to establish a characterization there  of the set of
generalized eigenfunctions in  an appropriately adapted Besov
space  generalizing parts of \cite{DS3}.   Principal tools include 
global solutions to the eikonal equation and   strong
radiation condition bounds.
\end{abstract}

\maketitle

\tableofcontents

\newpage

\section{Introduction}\label{sec:introduction} For a class of
long-range potentials we introduce a stationary scattering theory for
Schr\"odinger operators $H = -\Delta + V$ on $L^2({\mathbb
  R}^d)$ which
is regular at zero energy. In particular it is well defined at this
energy, and we use it to establish a characterization there  of the set of
generalized eigenfunctions in  an appropriately adapted Besov
space. The analogue of this characterization at positive energies for
potentials obeying $
\langle x \rangle^{\mu+|\alpha|} |\partial^{\alpha} V(x)|
\leq C_{\alpha}$ for some $\mu>0$ is well known \cite{AH,Ho2,GY}. It
goes as follows: 

For all $\lambda>0$
and all generalized eigenfunctions,  $(H-\lambda)u_\lambda=0$,  in the
Besov space $B(|x|)^*$ there exist unique
$\tau, \tilde \tau\in L^2(S^{d-1})$ such that 
\begin{align}\label{eq:asymPostive}
  u_\lambda(x)-C|x|^{-(d-1)/2}\parb{\e^{\i S(x,\lambda)}\tau(\omega)+\e^{-\i S(x,\lambda)}\tilde\tau(\omega)}\in B(|x|)_0^*.
\end{align} Here $ S(\cdot ,\lambda)={\sqrt
  \lambda}|x|+o(|x|)$ is a solution to the eikonal equation,
$\omega=x/|x|$  and $B(|x|)_0^*\subset B(|x|)^*$ are specified by
\begin{align*}
u\in B(|x|)^*&\Leftrightarrow u\in  L^2_{\rm loc}(\R^d) \mand \sup_{R>1}R^{-1}\|F(|x|<R)u\|<\infty,\\
  u\in B(|x|)_0^*&\Leftrightarrow u\in  L^2_{\rm loc}(\R^d) \mand \lim_{R\to \infty} R^{-1}\|F(|x|<R)u\|=0. 
\end{align*}
Moreover we can write  $\tilde
\tau(\omega)=(S(\lambda)^{-1}\tau)(-\omega)$ where the operator $S(\lambda)$
is a 
unitary operator on $L^2(S^{d-1})$ named the scattering matrix at
energy $\lambda$. The family of these  operators is
connected   to a scattering
operator from time-dependent scattering theory by a Legendre transformation. 

The (inverse) scattering matrix at energy $\lambda$ is determined by
\eqref{eq:asymPostive}: For all $\tau \in L^2(S^{d-1})$
there exist a unique $\tilde \tau\in L^2(S^{d-1})$ and a unique   generalized
eigenfunction $u_\lambda\in B(|x|)^*$ such that the asymptotics
\eqref{eq:asymPostive} is fulfilled. Whence indeed the set of
generalized eigenfunctions in $B(|x|)^*$ at any  positive energy
$\lambda$ is
characterized by \eqref{eq:asymPostive}. The variable $\omega$ may be
thought of as the observable asymptotic normalized velocity, see \cite{DS3} for
discussion.

We refer to \cite{Me,Va} for a related approach to stationary
scattering theory for a class of geometric models.

For a class of potentials, negative at infinity and to leading order
spherically symmetric, the above constructions were extended down to
(and including) zero energy \cite{DS3}. We refer to  \cite{DS2, Fr} for  explicit
calculations  of the scattering matrix at zero energy  and to \cite{Ya,SW} for
 related
one-dimensional results on asymptotics of scattering quantities. Whence in particular  the
set of generalized  
zero energy eigenfunctions in  an appropriately adapted Besov
space is characterized in \cite{DS3} for the restrictive class of potentials. Since
in turn this class of potentials is close to being optimal for the
existence in Classical Mechanics  of  asymptotic normalized velocity at zero
energy the given
characterization result may be viewed as ``best possible''. Nevertheless the
purpose of this paper is to provide a similar characterization of  generalized 
zero energy eigenfunctions for a bigger class of potentials
than considered in  \cite{DS3}. Again we obtain a 
parametrization by $L^2(S^{d-1})$ howewer the isomorphism is 
 different. Rather than involving functions on a sphere of  asymptotic
 normalized  velocities it will be in terms of functions on a sphere of {\it initial
   velocities}. In this sense our approach will be in the spirit of \cite{ACH} where
 a distorted Fourier transform is constructed for order zero
 potentials at high energies in terms of a family  of initially
 controlled geodesics. We prove low-energy  radiation condition bounds of
 independent interest.

The class of potentials to be studied in this paper is introduced in
Section \ref{sec:Smaller class}. In the remaining part of the present
section we review various background results for a somewhat bigger
class. The zero energy characterization problem makes
 sense for
this bigger class  (at least to some degree), see Subsection \ref{sec:Open problems}. Whence  the  class considered in the bulk of
the paper may not be optimal for the characterization problem
although  a further extension would involve
 difficult problems to overcome, see  Subsection \ref{sec:Idea of procedure}.

\subsection{A priori quantum bounds}
\label{sec:priori-quant-bounds}
\label{sec:Preliminary quantum bounds}
We give an account of some recent results \cite{Sk}. These
include  Besov space bounds of the resolvent at low  energies in any
dimension for a
class of potentials  that are negative  and obey a virial
condition with these conditions imposed  at infinity only. There are  two boundary values of the resolvent at
zero energy which are  characterized by radiation conditions. These
radiation conditions are zero energy versions of  the well-known
Sommerfeld radiation condition.

We consider  the
Schr\"odinger operator
$H = -\Delta + V$  on $\mathcal H=L^2({\mathbb R}^d),\, d\geq1$, where
the potential
$V$ obeys the following condition. We  
use the notation  $\langle x \rangle = \sqrt{ x^2+1}$,
$\N_0=\N\cup\{0\}$, and for $\mu \in (0,2)$ the  notation  $s_0=1/2+\mu/4$. 
\begin{cond}
  \label{cond:lap}
  Let $V=V_1+V_2$ be a real-valued function defined on ${\mathbb
    R}^d$; $ d\geq1$.  There exists $\mu \in (0,2)$ such that the
  following conditions {\rm
    \ref{it:assumption1}--\ref{it:assumption_last}} hold.

  \begin{enumerate}[\normalfont (1)]
  \item \label{it:assumption1} There exists $\epsilon_1 > 0$ such that
    $V_1(x) \leq -\epsilon_1 \langle x \rangle^{-\mu}.$
  \item \label{it:assumption2} $V_1\in C^\infty(\R^d)$. For all
    $\alpha\in \N_0^{d}$ there exists $C_{\alpha} >0$ such that
    \begin{equation*}
      \langle x \rangle^{\mu+|\alpha|} |\partial^{\alpha} V_1(x)|
      \leq C_{\alpha}.
    \end{equation*}
  \item \label{it:assumption3} There exists $\tilde \epsilon_1 > 0$
    such that $-|x|^{-2} \left(x\cdot \nabla (|x|^2 V_1)\right) \geq
    -\tilde \epsilon_1 V_1$.
   
  \item \label{it:assumption4} There exists $\delta,C,R > 0$ such that
    \begin{equation*}
      |V_2(x)| \leq C |x|^{-2s_0-\delta},
    \end{equation*}
    for $|x| > R$.
  \item \label{it:assumption_last} $V_2\in L^p_{\rm loc }(\R^d)$,
    where $p=2$ if $d=1,2,3$ and $p>d/2$ if $d\geq 4$.
  \end{enumerate}
\end{cond}
Due to \ref{it:assumption4} and \ref{it:assumption_last} the operator $V_2 (-\Delta +i)^{-1}$ is a compact operator on
     $L^2({\mathbb R}^d)$, cf. \cite[Theorem X.20]{RS}. Whence  $H$ is self-adjoint. The
Schr\"odinger operator with  an  attractive
Coulomb potential in
dimension $d\geq3$  is a particular example.

Let $\theta \in (0, \pi)$, $\lambda_0>0$  and define
\begin{equation}
  \label{eq:def_Gamma_theta}
\Gamma_{\theta,\lambda_0} = \{ z \in {\mathbb C}\setminus \{0\}
\,\big{ | } 
\arg z \in (0,\theta) , \,
|z| \leq \lambda_0\}.
\end{equation}

For a Hilbert space ${\mathcal H}$ (which in our case  will be
$L^2({\mathbb R}^d)$) we denote by ${\mathcal B}(\mathcal H)$ the
space of bounded linear operators on ${\mathcal H}$  
(a similar notation will be used for Banach spaces).
A ${\mathcal B}(\mathcal H)$-valued function $T(\cdot)$ on
$\Gamma_{\theta,\lambda_0}$ is said to be uniformly H\"{o}lder continuous in 
$\Gamma_{\theta,\lambda_0}$ if there exist $C,\gamma>0$ such that 
$$
\|T(z_1)-T(z_2) \|\leq C|z_1-z_2|^{\gamma}\; {\rm{for\; all}}\; z_1,z_2\in
\Gamma_{\theta,\lambda_0}.
$$

 We denote the resolvent of $H$ by $R(z) = (H -z)^{-1}$. The notation $B(|x|)$
and $B(|x|)^*$ refers to the Besov space for the operator of
multiplication by $|x|$ and its dual space, respectively.

\begin{proposition}[LAP]
\label{thm:lap}
Suppose Condition \ref{cond:lap}.
 For all  $s>s_0$
  the family of operators $T(z)=\langle x
\rangle^{-s}R(z) 
\langle x \rangle^{-s}$ is uniformly H\"{o}lder continuous in 
$\Gamma_{\theta,\lambda_0}$. 
In particular the 
limits
\begin{align*}
T(0 +\i0)&=\langle x \rangle^{-s} R(0 + \i0) 
\langle x \rangle^{-s}
= \lim_{\Gamma_{\theta,\lambda_0}\ni z \rightarrow 0}
T(z),\\
T(0 -\i0)&=\langle x \rangle^{-s} R(0 -\i0) \
\langle x \rangle^{-s}
= \lim_{\Gamma_{\theta,\lambda_0}\ni z \rightarrow 0}
T(\bar z)
\end{align*}
exist in
${\mathcal B}(L^2({\mathbb R}^d))$.

There exists
  $C>0$ such that for all $z \in \Gamma_{\theta,\lambda_0}$
\begin{equation}
  \label{eq:limitbound2j}
 \|
(|z|+\inp{x}^{-\mu})^{1/4}R(z)(|z|+\inp{x}^{-\mu})^{1/4}
\|_{\vB (B(|x|),B(|x|)^*)} \leq C.
\end{equation}
\end{proposition}
\subsubsection{Zero energy Sommerfeld radiation condition}\label{sec:Zero energy Sommerfeld radiation condition}
 We shall  give an outline  of some microlocal estimates
 and characterizations of  solutions to the equation $Hu=v$. In particular we estimate and
characterize the ``outgoing''  solution whose existence is provided by
Proposition 
\ref{thm:lap}. This particular
solution is given as follows in terms of Besov
spaces. First note that the  relevant Besov space at zero energy  is
$B^\mu:=B(\inp{x}^{2s_0})=\inp{x}^{-\mu/4}B(|x|)$, cf.
\eqref{eq:limitbound2j}. We have the following characterization of the
corresponding dual
space 
\begin{equation*}
  u\in (B^
\mu)^* \Leftrightarrow  u\in L^2_{\rm loc}(\R^d)\mand \sup_{R>1} R^{-s_0}\|F(|x|<R)u\|<\infty.
\end{equation*} A slightly smaller space is given by 
\begin{equation*}
  u\in (B^
\mu)^*_0 \Leftrightarrow u\in  L^2_{\rm loc}(\R^d) \mand \lim_{R\to \infty} R^{-s_0}\|F(|x|<R)u\|=0. 
\end{equation*}
Now suppose  $v\in  B^
\mu$.  Then due to Proposition 
\ref{thm:lap} there exists the weak-star limit
\begin{equation*}\label{eq:30o}
  u=R(0 + \i0) v
= \wslim_{\Gamma_{\theta,\lambda_0}\ni z \rightarrow 0}
R(z) v\in (B^
\mu)^*. 
\end{equation*}  Note that indeed this $u$ is a (distributional)
solution to the equation $Hu=v$. 

 To  state   microlocal properties  of  this solution we first introduce
 for all $\lambda\geq 0$ 
 the function 
\begin{equation}
  \label{eq:fsublambda}
  f=f_\lambda(x)=(\lambda+K\inp{x}^{-\mu})^{1/2};\, x\in \R^d,
\end{equation} where  $K:=\epsilon_1 \tilde \epsilon_1 /(2-\mu)$ with
the $\epsilon$'s given in  Condition
\ref{cond:lap}.
 In terms of  $f_0$ we  introduce symbols 
\begin{align}\label{eq:ab_0}
a_0= \frac{\xi^2}{f_0(x)^2},\;\;
b_0=  \frac{\xi}{f_0(x)} \cdot \frac{x}{\langle x \rangle},
\end{align} and we have, using here  Weyl quantization, 
\begin{subequations}
\begin{equation}\label{eq:31h}
  \Opw(\chi_-(a_0)\tilde \chi_-(b_0))u\in (B^
\mu)^*_0\mforall \chi_-\in C^\infty_\c(\R)\mand \tilde \chi_-\in C^\infty_\c((-\infty, 1)).
\end{equation}  

These estimates are  accompanied by ``high energy
estimates'', stated as follows: Let us note that 
\begin{equation*}\label{eq:25b}
  f_{|z|}^{-2}(x)\big |V_1(x)-z\big |\leq C_0':=\max(C_0/K,1),
\end{equation*}  where $C_0$ is given in Condition \ref{cond:lap}
 \ref{it:assumption2}
(i.e. the constant with $\alpha=0$). Consider  real-valued $\chi_-\in C^\infty_\c(\R)$  such that
$\chi_-(t)=1$  in a neighbourhood of $[0,C_0']$,
and let $\chi_+:=1-\chi_-$. For all
such functions $\chi_+$
\begin{equation}\label{eq:31hh}
  \Opw(\chi_+(a_0))u\in (B^
\mu)^*_0.
\end{equation}
 \end{subequations}

The support property of $\tilde \chi_-$ in \eqref{eq:31h}  mirrors that
the particular solution  studied  is the outgoing one, and we refer to
\eqref{eq:31h}  as  the  
{\it outgoing Sommerfeld radiation condition}. This condition   (in fact a
weaker version) suffices   for a characterization as expresssed in the
following result. Here and henceforth $L^2_m:=\inp{x}^{-m}L^2(\R^d)$. 
\begin{proposition}[Uniqueness of outgoing
  solution, data in $B^
\mu$]\label{thm:somm-radi-condss}  Suppose   $v\in  B^
\mu$. Suppose $u$ is a  distributional 
solution to the equation $Hu=v$ belonging to the space $
L^2_m$ for some $m\in \R$, and suppose that there
there exists $\kappa\in(0,1]$ such that 
\begin{equation}\label{eq:31hf}
  \Opw(\chi_-(a_0)\tilde \chi_-(b_0))u\in (B^
\mu)^*_0\mforall \chi_-\in C^\infty_\c(\R)\mand \tilde \chi_-\in C^\infty_\c((-\infty, \kappa)).
\end{equation} Then $u=R(0 + \i0) v$. In particular \eqref{eq:31h} and
\eqref{eq:31hh} hold.
\end{proposition}

The ``incoming'' solution  $u=R(0
- \i0) v$ can be   characterized  similarly. These results generalize
\cite[Proposition 4.10] {DS3} at zero energy. 
\begin{remark}\label{remark:zero-energy-somm} There are similar
  results for positive energies. For $R(\lambda+\i 0)$ we have the
  same conclusion $u=R(\lambda+\i 0)v$ for an outgoing solution to
  $(H-\lambda)u=v$. This means more precisely that  if we replace the Besov spaces
   by replacing $s_0\to s_0=1/2$ in the definition of these spaces and change the localization symbols
  $a_0,b_0$ in \eqref{eq:ab_0} and  \eqref{eq:31hf} by replacing
  $f_0\to f_\lambda$ there, then indeed the solution $u$ is given by
  $u=R(\lambda+\i 0)v$.  This result is known for  larger classes of potentials, see
\cite[Theorem 30.2.10] {Ho2} and \cite{GY}.
\end{remark}

\subsection{Open problems}\label{sec:Open problems}
Define
  under Condition \ref{cond:lap} the operator 
  \begin{equation*}
    \delta(0)=(2\pi\i)^{-1}\parb{R(0
+ \i0)-R(0- \i0)}=\pi^{-1}\Im \parb{R(0
+ \i0)}\in \vB(B^\mu,(B^\mu)^*), 
  \end{equation*} and note that its range 
  \begin{equation*}
    \Ran(\delta(0))\subseteq \vE_0:=\{ u\in (B^\mu)^*|\,Hu=0\}.
  \end{equation*}
Under some stronger conditions it follows from \cite[Theorem
8.2]{DS3} that $\Ran(\delta(0))= \vE_0$ (proved in terms of wave
matrices at zero energy).  Equality and characterization of $\vE_0$
are  open problems  under Condition
\ref{cond:lap}, in fact  it is only known  that $\delta(0)\neq
0$, see  \cite{FS}. More specifically ``scattering theory at zero energy'' in the
spirit of \cite[Theorem
8.2]{DS3} is an open problem under Condition
\ref{cond:lap}. In this paper we address these problems for an
intermediate class  of potentials, i.e. a smaller
class  than the one defined by Condition \ref{cond:lap} but bigger than the
one studied in \cite{DS3}.

\subsection{Ideas of procedure and results}\label{sec:Idea of procedure} Let us
give an outline of a possible procedure for solving the problems 
posed in the proceeding subsection. This procedure will  be implemented for the subclass of potentials
to be introduced in Section \ref{sec:Smaller class}. The corresponding (main) results
are  stated more precisely  in  Theorem \ref{thm:char-gener-eigenf-1}. For simplicity we
assume in the discussion below that $V$ is negative.

First we need a global solution to the eikonal equation (or at least solving outside a compact set) 
\begin{align*}
  |\nabla_x S(x,\lambda)|^2=\lambda-V(x);\,\lambda\geq 0.
\end{align*} The existence for  $\lambda=0$ is not known under Condition
\ref{cond:lap}. Potentially we could define
$S(\cdot,\lambda)$ to be the  distance in the metric $g_\lambda=(\lambda-V)\d
x^2$   to the origin in $\R^d$,
i.e. $S(x,\lambda)=d_{g_\lambda}(x,0)$. This is the so-called maximal
solution to the eikonal equation. However  under Condition
\ref{cond:lap} it is a problematic choice, in fact for
$d\geq 2$ it might be expected  that  in some generic sense  this $S(\cdot,\lambda)\notin C^1(\R^d\setminus\{0\})$. 

However for the
subclass of potentials to be considered the above geometric
construction is  manageable and we shall consider the
corresponding geodesic flow
\begin{align*}
  \tfrac{\d}{\d s}\Phi&=(\lambda-V(\Phi))^{-1}\nabla_x
    S(\Phi,\lambda),\,\,\Phi(0,\sigma)=0,\,\tfrac {\d}{\d s}\Phi(s,\sigma)_{|s=0}=(\lambda-V(0))^{-1/2}\sigma;\\&(s,\sigma)\in[0,\infty)\times S^{d-1}.
\end{align*} In particular it turns out that this flow is  a diffeomorphism $\Phi:
\R_+\times S^{d-1}\to \R^d\setminus\{0\}$. 

Next for an appropriate Jacobian type function $J$, see
\eqref{eq:diag},  we propose to introduce
\begin{equation}
  \label{eq:diag2}
  F^+(\lambda)v=\vGlim_{s\to \infty}\,\parb{J^{1/2}\e^{-\i S(\cdot,
      \lambda)} R(\lambda +\i 0)v}\parb{\Phi (s,\cdot)},
\end{equation} where $\vG:=L^2(S^{d-1},\d \sigma)$. This is for $v$ in
an appropriate dense subset of $L^2(\R^d)$, and using  an integration by
parts  and Stone's formula we then derive the following formula for the
orthogonal projection onto the continuous subspace of $H$:
\begin{align*}
  \|P_\c v\|^2=\int _0^\infty  \|F^+(\lambda)v\|_{\vG}^2\,\d \lambda.
\end{align*} This leads to  the  {\it distorted Fourier transform}
\begin{align*}
  F^+:=\int _0^\infty \oplus F^+(\lambda)\,\d \lambda.
\end{align*} This map is a partial isometry diagonalizing  $H_\c$,
i.e. $F^+H_\c=M_\lambda F^+$.
 We  show the existence of the limit
 \eqref{eq:diag2} by using some new low-energy  radiation condition
 bounds valid under the conditions of
 Section \ref{sec:Smaller class}. The reader may consult
 \eqref{eq:extbF2} for a  somewhat cleaner definition.

Now we can address the problems of Subsection \ref{sec:Open problems}
(under  these conditions).
Indeed 
\begin{equation*}
    \Ran(\delta(0))=\vE_0=\{ u\in (B^\mu)^*|\,Hu=0\}
  \end{equation*} follows from the following properties:
  \begin{align}
&F^+(0):B^\mu\to \vG\text{ is onto},\nonumber\\
   &F^+(0)^* :\vG\to \vE_0\text{ is a bi-continuous  isomorphism},\label{eq:para}\\
&\delta(0)=F^+(0)^* F^+(0).\nonumber
  \end{align} Furthermore note that \eqref{eq:para} constitutes  a
  parametrization of $\vE_0$. The isomorphism 
  $F^+(0)^*$,  named the {\it  
    wave  matrix at zero
      energy},  is given more explicitly as follows:
  For any $u=2\pi\i F^+(0)^* \tau\in\vE_0$
    \begin{align}\label{eq:gen1aa}
      u(x)-J^{-1/2}(x)\parb{\e^{\i S(x,0)}\tau(\sigma)-\e^{-\i
        S(x,0)}\tilde\tau(\sigma)}\in
(B^\mu)_0^*;\,x=\Phi(t,\sigma).
    \end{align} 

The function
    $\tilde\tau\in\vG$ in \eqref{eq:gen1aa} is uniquely determined
    from $u$ 
    and it is of the form $\tilde \tau=S(0)^{-1}\tau$ where $S(0)$ is a unitary operator
    on $\vG$. This  operator  is called the {\it  scattering matrix at zero
      energy}. Combined with  similar constructions for $\lambda>0$
     the scattering matrix $S(\lambda)$ is strongly continuous in
    $\lambda\geq 0$. Whence this {\it renormalized} stationary
    scattering theory is {\it regular 
    at zero
      energy}.

\section{Class of potentials}\label{sec:Smaller class}
We introduce the class of potentials to be studied in this paper.  The zero
energy dynamics for this   class of potentials is generically 
qualitatively very different (unless $d=1$) from the one  for potentials in the
smaller   class  of \cite{DS3}. We
give an example to that effect.
\subsection{Conditions}\label{sec:Conditions}
 \begin{cond}[Unperturbed potential]
\label{cond:rad}
Let $V=V_1+V_2$ be a
real-valued function defined on  ${\mathbb R}^d$; $
d\geq1$. There exists  $\mu \in (0,2)$  such
that the following  conditions
{\rm
\ref{it:assumption1r}--\ref{it:assumption_lastr}} hold.

\begin{enumerate}[\normalfont (1)]
   \item \label{it:assumption1r} There exists $\epsilon_1 > 0$ such
     that $V_1(x) \leq -\epsilon_1
\langle x \rangle^{-\mu}.$
   \item \label{it:assumption2r}     $V_1\in C^\infty(\R^d)$. For all $\alpha\in \N_0^{d}$
there exists
$C_{\alpha} >0$ such that
$$
\langle x \rangle^{\mu+|\alpha|} |\partial^{\alpha} V_1(x)|
\leq C_{\alpha}.$$
\item \label{it:assumption3r} 
     $V_1(x)=V_\rad(|x|)$ is spherically symmetric, and  there exists
     $\tilde \epsilon_1 > 0$ such that
     \begin{equation*}
       -2V_\rad(r) -rV_\rad'(r)
\geq -\tilde \epsilon_1 V_\rad(r).
     \end{equation*}
   
   \item \label{it:assumption_lastr} $V_2$ is compactly supported, and  
$V_2\in L^p(\R^d)$,  where $p=2$ if $d=1,2,3$ and  $p>d/2$ if  $d\geq 4$.
\end{enumerate}
\end{cond}
 
Given Condition \ref{cond:rad} we consider the  class $\vW$
of real-valued
smooth functions $W$ on $\R^d$  obeying that for all
$\alpha\in \N_0^{d}$
\begin{subequations}
\begin{equation}\label{eq:Wbnds1}
  \sup_{x\in \R^d} \langle x \rangle^{\mu+|\alpha|} |\partial^{\alpha} W(x)|<\infty.
\end{equation} Given $l\in
\N$   we say 
that $W_\epsilon\in \vW$ is $\epsilon$-\emph {small} 
if  for some $\epsilon>0$
\begin{equation}\label{eq:Wbnds2}
  \max_{|\alpha|\leq l}\sup_{x\in \R^d} \langle x
  \rangle^{\mu+|\alpha|} |\partial^{\alpha} W_\epsilon(x)|\leq \epsilon.
\end{equation}  
\end{subequations}
Clearly this  quantity  depends on the given $l$, however we prefer
for  the
above 
teminology of $\epsilon$-smallness to suppress this dependence. If
in a given context $l$ is not specified, it is tacitly understood that
$l=2$ (although for example $l=1$  suffices   for Proposition
\ref{prop:resol_bounds}). We use $l=4$ in Lemma
\ref{lemma:diagonalization} stated below. Similarly we need $l\geq 4$
in Lemma \ref{lemma:diagonalizationb} and Proposition
\ref{Prop:radi-cond-bounds} (a sufficient choice  $l=l(\mu,d)$ can be
calculated, however we shall not bother). Consequently our main result
Theorem \ref{thm:char-gener-eigenf-1}  will depend on some fixed  $l\geq 4$
in the definition (\ref{eq:Wbnds2}) of $\epsilon$-small perturbations.

We shall study potentials of the form
$V_\epsilon=V+W_\epsilon$ where $V=V_\rad+V_2$ agrees with Condition
\ref{cond:rad} and $W_\epsilon\in \vW$ is an $\epsilon$-small
perturbation. The class of such potentials $V_\epsilon$, say $\vV_\epsilon$, is a particular
subclass of  the one defined by Condition \ref{cond:lap}; here we need $\epsilon$
small. In fact at various other points of the paper we need to take $\epsilon>0$ small,
however this will be  expressible in terms of $V_\rad$ only, which
henceforth is considered as fixed. For convenience we assume throughout
the paper that
\begin{subequations}
  \begin{equation}\label{eq:supp_rad}
  V_\rad(r)=V_\rad(0)\mfor r\leq R:=(-V_\rad(0))^{-1/2},
\end{equation} and similarly for perturbations that 
\begin{equation}\label{eq:supp_pert}
  W_\epsilon(x)=0\mfor |x|\leq R.
\end{equation} 
\end{subequations} We can freely assume \eqref{eq:supp_rad} and
\eqref{eq:supp_pert}. As for $V_\rad$ the property  \eqref{eq:supp_rad} can be assumed  possibly  upon making
$\epsilon_1$ smaller (but not changing  $\tilde \epsilon_1$) and 
changing $V_2$. Although this is an elementary fact it is not
completely obvious. Let us give a proof: Decompose $1=\chi_++\chi_-$
where $\chi_+,\chi_-\in C^\infty(\R_+)$ are monotone, $\chi_+(t)=1$ for $t\geq 2$
and $\chi_+(t)=0$ for $t\leq 1$. Introduce
\begin{align}\label{eq:trun}
  V_n(r)=V_\rad(r)\chi_+(r/n)-n^{-2}\chi_-(r/n);\;n\in\N.
\end{align} We claim that for any $n$ big enough such that
$\epsilon_1\inp{2n}^{-\mu}\geq n^{-2}$ indeed  Conditions
\ref{cond:rad}\ref{it:assumption1r}--\ref{it:assumption3r} and
\eqref{eq:supp_rad} hold with $\epsilon_1$ replaced by $n^{-2}$, new
constants $C_\alpha$, the
same constant $\tilde \epsilon_1$ and with $R=n$, respectively. To see
that  indeed the same  $\tilde \epsilon_1$ works we consider the
estimates
\begin{align*}
  -rV_n'(r)&\geq (2-\tilde \epsilon_1)V_n(r)+(2-\tilde
  \epsilon_1)n^{-2}\chi_-(r/n)-\tfrac rn
  \chi_-'(r/n)\parb{-V_\rad(r)-n^{-2}}\\
&\geq (2-\tilde \epsilon_1)V_n(r)-\tfrac rn
  \chi_-'(r/n)\parb{\epsilon_1\inp{2n}^{-\mu}-n^{-2}}\\&\geq (2-\tilde \epsilon_1)V_n(r);
\end{align*} we assumed that $\tilde
  \epsilon_1\leq 2$. The other statements are obvious.
Similarly \eqref{eq:supp_pert}  can be assumed by
changing $V_2$ and  possibly by taking $\epsilon$
smaller.

\subsubsection{Example}\label{exa:Example} Let $g\in C^\infty(\R)$ be
$2\pi$-periodic with $\max g'\geq 1-\mu/2$. Let $\chi\in C^\infty(\R_+)$
obey $\chi(r)=0$ for $r<1$ and $\chi(r)=1$ for $r>2$. Similarly
introduce 
for $\mu\in (0,2)$  and (large) $n\in \N$  a function 
$h=h_n\in C^\infty(\R_+)$ obeying 
\begin{align*}
  \begin{cases}
    h(r)=r/n&\text{ for } r\leq n\\
h(r)=(1-\mu/2)^{-1}r^{1-\mu/2}+C_n&\text{ for } r\geq 2n\\
h'(r)> \max(0,-rh''(r))&\text{ for } r>0
  \end{cases}\quad.
\end{align*} Note that the construction \eqref{eq:trun}
with $V_\rad(r)=-r^{-\mu}$ leads to the particular example  $h_n(r)=\int^r_0 \sqrt{-V_n(t)}\d
t$. We construct in dimension $d=2$ a potential  in
terms of a parameter $\epsilon\geq 0$ and polar coordinates
$(r,\theta)$ (i.e. $x=(r\cos \theta,r\sin \theta)$):
\begin{align*}
  S_\epsilon(x,\lambda=0)&=h_n(r)\exp\{\epsilon  g(\theta-\epsilon \ln r)\chi(r/n)\},\\
V_\epsilon(x)&=- |\nabla S_\epsilon(x,\lambda=0)|^2.
\end{align*}
Clearly $V_{\epsilon=0}(x)=V_{\rad }(r)$ obeys Condition
\ref{cond:rad} and \eqref{eq:supp_rad} (the latter
with $R=n$). Moreover clearly $W_\epsilon(x):=V_\epsilon(x)-V_{\rad
}(r)$ satisfies \eqref{eq:Wbnds1} and \eqref{eq:supp_pert}. Morever for any $l\in\N$ there exists $C>0$,  
sufficiently large and possibly depending on $n$, such that  the
potential $W_\epsilon$ is $(C\epsilon)$-small. So up to a linear 
reparametrization  also  \eqref{eq:Wbnds2} is satisfied.

This example does not fit into the framework of  \cite{DS3}. In fact
for the class studied in \cite{DS3} classical zero energy scattering
orbits  have asymptotic normalized velocities. This would for the
above example mean that there exist $\lim_{t\to \pm
  \infty}\theta(t)$. However this cannot be as the following arguments
show:  Consider the flow (in polar coordinates)
\begin{subequations}
\begin{align}
  \label{eq:flowd=2}
  \begin{cases}
    \dot r(=\tfrac {\d}{\d s}r(s))&=(-V_\epsilon(x))^{-1}\partial_r S_\epsilon(x,\lambda=0)\\
\dot \theta(=\tfrac {\d}{\d s}\theta(s))&=(-V_\epsilon(x)r^2)^{-1}\partial_\theta
S_\epsilon(x,\lambda=0)\\
(r,\theta)(s=1)&=(n,\sigma).
  \end{cases}\quad. 
\end{align} 
Noticing that for $\epsilon>0$ small $\partial_r S_\epsilon>0$  we can consider
$\theta$ as a function of $r$ determined by the single equation
\begin{equation}
  \label{eq:singleODE}
  \tfrac {\d \theta}{\d r}=F(r,\theta):=\tfrac \epsilon r
\frac{g'\chi}{rh'/h+\epsilon \tfrac rn  g \chi'-\epsilon^2 g' \chi}.
\end{equation} Here of course $g$  and $\chi$ are   functions of
$\psi:=\theta-\epsilon \ln r$ and $r/n$, respectively. For $r\geq 2n$
\eqref{eq:singleODE} reduces to 
\begin{equation*}
  \label{eq:singleODE2}
  \tfrac {\d \psi}{\d r}=\tfrac \epsilon r
\parbb{\frac{g'}{1-\tfrac \mu 2-\epsilon^2 g'}+O(r^{\mu/2-1})-1}=\tfrac \epsilon r\parbb{
\frac{(1+\epsilon^2)g'-1+\tfrac \mu 2}{1-\tfrac \mu 2-\epsilon^2 g'}+O(r^{\mu/2-1})}.
\end{equation*} Introducing a new time, $\d \tau/\d r =\epsilon
r^{-1}$, we obtain 
\begin{equation}
  \label{eq:singleODE3}
  \tfrac {\d \psi}{\d \tau}=
\frac{(1+\epsilon^2)g'(\psi)-1+\tfrac \mu 2}{1-\tfrac \mu 2-\epsilon^2 g'(\psi)}+O\parb{\e^{(\mu/2-1)\tau/\epsilon}}.
\end{equation} 
\end{subequations} Note that to leading order  (\ref{eq:singleODE3})
is autonomous.   Any solution $\psi$ to \eqref{eq:singleODE3} converges to a  root of
the corresponding fixed point equation $g'(\psi)=(1-\mu/2)(1+\epsilon^2)^{-1}$,
 say 
 $\psi\to \psi_0$.
In particular going back to the  time $s$ of
\eqref{eq:flowd=2} we conclude  that $\theta-\epsilon \ln r\to  \psi_0$ as
$s\to \infty$, and since $\ln r\to  \infty$ indeed also $\theta \to
\infty$. So  the asymptotic normalized velocity does not exist for the
flow (\ref{eq:flowd=2}).  Noticing   that \eqref{eq:flowd=2} defines a class of zero
energy scattering orbits in a reparametrized time we conclude that
indeed these orbits do not
have asymptotic normalized velocity. 

In Subsection
\ref{subsec:Geometric properties}  we study a flow of the type 
\eqref{eq:flowd=2}  
for general $\epsilon$-small perturbations in any dimension  
 (extended as well to any non-negative energy).

\section{Eikonal equation}\label{sec:Eikonal equation} 
 One reason for considering $\vV_\epsilon$ with $\epsilon$ small only
 is that  Classical Mechanics is particularly nice for this
 class. Whence (cf. \cite{CS}) there exists 
 a global solution to the eikonal equation 
 \begin{align}
   \label{eq:eikonal}
   \MoveEqLeft |\nabla
   S_\epsilon|^2=K_\epsilon;\\&K_\epsilon(x)=K_\epsilon(x,\lambda):=\lambda-V_\rad(|x|)-W_\epsilon(x),\;
   \lambda\geq 0.\nonumber
 \end{align} We also introduce
 \begin{align*}
   K_0(r)&=\lambda-V_\rad(r),\\f(r,\lambda)&=\sqrt{K_0(r)} \\
 S_0(x)&=S_0(|x|)=\int_0^{|x|}f(r,\lambda)\d r.
 \end{align*} As used in \cite{CS} we have uniformly in $
 r,\lambda\geq0$ 
 \begin{equation*}
   crf(r,\lambda)\leq S_0(r)\leq Crf(r,\lambda).
 \end{equation*} Notice also that $S_0$ is a solution to
 \eqref{eq:eikonal} if $W_\epsilon=0$.

Due to  \cite{CS} we have 
\begin{proposition}\label{prop:eikonal-equation} Let $V_\rad$ be given
  as in Condition \ref{cond:rad} (assuming also
  \eqref{eq:supp_rad}) and let $l\geq 2$. There exists
  $\epsilon_0>0$ such that for all $\epsilon\in(0,\epsilon_0]$ and all  
  $\epsilon$-{small} perturbations $W_\epsilon$ (assuming \eqref{eq:supp_pert}) there exists a family
  of real-valued smooth functions 
  $\{S_\epsilon\in C^\infty (\R^d\setminus\{0\})|\,\lambda\geq 0\}$ with the following properties:
\begin{enumerate}[\normalfont (1)]
\item\label{item:3} 
$|\nabla S_\epsilon(x)|^2=K_\epsilon(x) \mfor x\in
\R^d\setminus\{0\}.$
  \item \label{item:5} $S_\epsilon(x)=S_0(x)=f(0,\lambda)|x| $ for  $r=|x|\leq R=(-V_\rad(0))^{-1/2}$.
  \item \label{item:4} For all $r_0>0$, uniformly in $W_\epsilon$ with  $\epsilon\in
    (0,\epsilon_0]$ 
    \begin{equation*}
      {\max_{|\alpha|\leq l}\sup_{\lambda\geq 0}\sup_{|x|\geq r_0}}\langle x\rangle^{|\alpha|}\left|  S_0(x)^{-1}\partial_x^\alpha
        S_\epsilon(x)\right|{<\infty}.
    \end{equation*}
\item \label{item:6}{Uniformly in $W_\epsilon$, 
   $\lambda \geq 0$ and $x\in \R^d\setminus\{0\}$} 
 \begin{align*}
   S_\epsilon(x)&=S_0(r)\parb{1+{O(\epsilon)}},\\
 \nabla S_\epsilon(x)&=f(r,\lambda)\;\parb{\langle \hat
   x|+{{O\parb{\epsilon^{3/4}}}}};\;\hat x:=x/r,\\
\nabla^2
S_\epsilon(x)&=\tfrac{f(r,\lambda)}{r}\;\parb{P_\perp+\tfrac{rf'(r,\lambda)}{f(r,\lambda)}P+{O\parb{\epsilon^{1/2}}}};\\&
\,\,\,\,\,\,\,\,\,\,\,\,\,\,\,\,\,\,\,\,\,\,\,\,\,\,\,\,\,\,\,\,\,\,\,\,\,\,\,\,\,\,\,\,\,P=P(\hat
x):=|\hat x\rangle\langle \hat x|,\,P_\perp:= I-P.
\end{align*}
\item \label{item:6b} For all $\alpha\in \N_0^d$
\begin{align*}
  \partial_x^\alpha
        S_\epsilon \in C (\R^d\setminus\{0\}\times [0,\infty));\;S_\epsilon =S_\epsilon (x,\lambda).
\end{align*}

\end{enumerate}
\end{proposition} We remark that  for $l=2$ the bounds \ref{item:4}
follow from \ref{item:6}. Having $l>2$
influences only  on \ref{item:4} and requires, according to
the proposition, an $\epsilon_0>0$ possibly depending on $l$. It is
tempting  to conjecture  that one could take  $l=2$ in the proposition and replace the
constraint  of \ref{item:4}, $|\alpha|\leq l$, by  $|\alpha|\leq k$
where $k$ is arbitrarily given. The new bounds would be uniform in 
perturbations  from any  bounded family (bounded in terms of the
seminorms \eqref{eq:Wbnds1}). However this is an open problem, and  in fact 
it is not  
known  whether
$\epsilon_0>0$ can be chosen independently of $l$, although there are
weaker estimates than  \ref{item:4} indeed independent of $l$, cf.  \cite[Proposition
1.2]{CS}. The latter  deficiency gives rise to a slight complication when
dealing with $S_\epsilon$ in the context of pseudodifferential
operators, see \eqref{eq:calc3}.

\subsection{Geometric properties}\label{subsec:Geometric properties} 
The construction of the function $S_\epsilon$ of  Proposition
\ref{prop:eikonal-equation}  is given  by a geometric procedure:
We consider the metric $g_\epsilon=K_\epsilon\d x^2$ on the manifold
$M=\R^d$ and the origin $o=0\in M$. Then  for all $x\in M$ the
number  $S_\epsilon(x)$ is the distance in this metric to $o$,
i.e. $S_\epsilon(x)=d_{g_\epsilon}(x,o)$. The function $S_\epsilon$ is
called the  \emph{maximal} solution to the eikonal equation.
\subsubsection{Flow}
In the metric $g_\epsilon$ the
unit-sphere in the tangent space $TM_o$ at the origin $o=0$
is given by $f(0,\lambda)^{-1}S^{d-1}$ where $S^{d-1}$ is the standard
unit-sphere in $\R^d$. We shall use the notation $\sigma$ for generic
points of $S^{d-1}$ and we let $\d \sigma$ denote the standard 
Euclidean surface measure on $S^{d-1}$. The exponential mapping at the
origin for the metric $g_\epsilon$ defines a diffeomorphism $\Phi:
\R_+\times S^{d-1}\to \R^d\setminus\{0\}$
\begin{equation*}
  \Phi(s, \sigma)=\exp_o(sf(0,\lambda)^{-1} \sigma),
\end{equation*} and we have the flow property
\begin{subequations}
  \begin{equation}
  \label{eq:flow}
  \tfrac{\d}{\d s}\Phi=\parb{K_\epsilon^{-1}\nabla
    S_\epsilon}(\Phi);\;s>0,\sigma\in S^{d-1}.
\end{equation} Since by assumption, cf.  \eqref{eq:supp_rad} and
\eqref{eq:supp_pert}, the conformal factor $K_\epsilon$ is constant for
$r=|x|\leq R=(-V_\rad(0))^{-1/2}$ we have explicitly 
\begin{equation*}
  \Phi(s, \sigma)=sf(0,\lambda)^{-1} \sigma\mfor s\leq 1.
\end{equation*} Whence  we can supplement \eqref{eq:flow} by the
``initial condition''
\begin{equation}
  \label{eq:ini}
  \Phi(1, \sigma)=f(0,\lambda)^{-1}\sigma.
\end{equation} 
\end{subequations}
The assertion above that $\Phi$ is a diffeomorphism
can be proved taking \eqref{eq:flow} and \eqref{eq:ini} as a
definition of the map. Notice   the  consequences of
\eqref{eq:eikonal}, \eqref{eq:flow}  and \eqref{eq:ini} that the distance 
$d_{g_\epsilon}(x,o)=S_\epsilon(x)=s$. This point of view is taken in the proof of
an  
analogous
statement \cite[Proposition 2.2]{ACH}. However the mapping  property  can also
be viewed as  an
independent part of the proof of Proposition
\ref{prop:eikonal-equation} given in \cite{CS}. The flow $\Phi$
constitutes  a
family of reparametrized Hamiltonian orbits for the Hamiltonian
\begin{equation}
  \label{eq:clasHam}
 h_\epsilon=\xi^2+V_\rad(|x|)+W_\epsilon(x) 
\end{equation}
 at energy $\lambda$. It is continuous in $\lambda$, i.e. $\Phi\in
 C (\R_+\times S^{d-1}\times [0,\infty))$; $\Phi=\Phi(s,\sigma,\lambda)$.
\subsubsection{Surface measure} The mapping $\Phi
(s,\cdot):S^{d-1}\to \vS_\epsilon(s):=\{x\in \R^d|S_\epsilon(x)=s\}$ induces a  measure on $S^{d-1}$ by
pullback $\d \omega=\Phi (s,\cdot)^*\d A(x)$ where $\d A(x)$ refers to
the  Euclidean surface  measure on $\vS_\epsilon(s)$. A computation using
\eqref{eq:flow} and \eqref{eq:ini} shows that
explicitly
\begin{align}
\begin{split}
  \label{eq:surf_meas}
\d  \omega&=K_\epsilon^{1/2}(x)m_\epsilon(x)\,\d
\sigma;\\ m_\epsilon(x)&=f(0,\lambda)^{2-d}K_\epsilon^{-1}(x)\exp\parbb{\int_1^s\parb{K_\epsilon^{-1}\triangle
    S_\epsilon} \parb{\Phi (t,\sigma)}\, \d t },\;x=\Phi (s,\sigma).
\end{split}
\end{align} Indeed, take
local coordinates $\theta_1,\dots,\theta_{d-1}$ on $S^{d-1}$, write
\eqref{eq:flow} as $\dot \eta=F(\eta)$ and let $A$ be the $d\times
(d-1)-$matrix with entries $a_{ki}
=\partial_{\theta_i}\eta^k$. The pullback $\d \omega$ is computed
from the metric $g_{ij}=(A^T A)_{ij}$ noting that the  determinant $|g|$ obeys
\begin{equation*}
  \tfrac {\d}{\d s}\ln |g|=\tr\parb{ (A^T A)^{-1}\tfrac {\d}{\d s}(A^T A)}=\tr\parb{(B^T+B)P}=2K_\epsilon^{-1}\triangle
    S_\epsilon-\tfrac{\d}{\d s}\ln K_\epsilon(\Phi),
\end{equation*} where $B=F'$ (the Jacobian matrix) and
$P_{kl}=\delta_{kl}-(\partial_kS_\epsilon)(\partial_lS_\epsilon)K_\epsilon^{-1}$. We
integrate and obtain
\begin{equation*}
\d  \omega =|g|^{1/2}\d
\theta=f(0,\lambda)K_\epsilon^{-1/2}(x)\exp\parbb{\int_1^s\parb{K_\epsilon^{-1}\triangle
    S_\epsilon} \parb{\Phi (t,\sigma)}\d t } f(0,\lambda)^{1-d} \,\d
\sigma,
\end{equation*} showing \eqref{eq:surf_meas}.

\subsubsection{Volume measure}  In combination with \eqref{eq:surf_meas} the co-area formula, cf. \cite[Theorem
C.5]{Ev}, yields for
(reasonable) functions $\phi$ on $\R^d$

\begin{align}
\begin{split}
  \label{eq:co_area}
  \int \phi(x) \,\d x&=\int_0^\infty \d s \int_{\vS_\epsilon(s)}\phi
  K_\epsilon^{-1/2}\,\d A(x)\\&=\int_0^\infty \d s
  \int_{S^{d-1}}\parb{\phi m_\epsilon}\parb{\Phi (s,\cdot)}\, \d
\sigma.
\end{split}
\end{align} 

Let $\vB_\epsilon(s):=\{x\in
\R^d|S_\epsilon(x)\leq s\}$ for $s>0$. Clearly $\partial
\vB_\epsilon(s)=\vS_\epsilon(s)$ and whence the Gauss integration theorem, cf. \cite[Theorem
C.1]{Ev}, yields for $j=1,\dots,d$
\begin{align}
  \begin{split}
  \label{eq:co_Gauss}
  \int_{\vB_\epsilon(s)} (\partial_j\phi)(x)
  dx&=\int_{\vS_\epsilon(s)}\phi(\partial_jS_\epsilon)  K_\epsilon^{-1/2}\,\d A(x)\\&=\int_{S^{d-1}}\parb{\phi(\partial_jS_\epsilon) m_\epsilon}\parb{\Phi (s,\cdot)}\,\d
\sigma.
  \end{split}
\end{align} 

\subsection{Diagonalization} \label{subsec:Diagonalization}
  Under the conditions of Section \ref{sec:Smaller class} we consider
  the Hamiltonian  $H = -\Delta + V_\epsilon$  on $\mathcal
  H=L^2({\mathbb R}^d) $. Denoting the corresponding continuous part
  by  $H_\c$ we aim at constructing a
diagonalizing transform taking $H_\c\to M_\lambda$ where $M_\lambda$ is
multiplication by $\lambda$ in $\widetilde\vH :=L^2(\R_+,\d \lambda;\vG)$ with
$\vG:=L^2(S^{d-1},\d \sigma)$. Here we explain our  procedure leaving 
the details of implementation to Section \ref{sec:Distorted Fourier
  transform}. It goes as follows, assuming
below $v\in L^2_3$ (recall $L^2_m:=\inp{x}^{-m}L^2(\R^d)$):  By
Stone's formula, cf. \cite{RS},
\begin{align*}
  \|P_\c v\|^2=\pi^{-1}\lim_{\lambda_0\to \infty }\int _0^{\lambda_0} \inp*{v, (\Im R(\lambda +\i
    0))v}\,\d \lambda=\pi^{-1}\int _0^\infty \inp*{v, (\Im R(\lambda +\i
    0))v}\,\d \lambda.
\end{align*} Whence writing $u=R(\lambda +\i 0)v$, $p_j=-\i \partial_j$ and using
\eqref{eq:co_Gauss}
\begin{align*}
  \|P_\c v\|^2&=\pi^{-1}\int _0^\infty  \Im \inp*{(H-\lambda)u, u}\,\d
  \lambda\\
&=\pi^{-1}\int _0^\infty  \lim_{s\to \infty}\Re \sum^d_{j=1}\int_{S^{d-1}}
\parb{\overline{(p_ju)}u (\partial_jS_\epsilon)m_\epsilon}\parb{\Phi (s,\cdot)}\,\d \sigma\d \lambda.
\end{align*} 
Next we substitute
$p_ju=(p_j-\partial_jS_\epsilon)u+(\partial_jS_\epsilon )u$. The
contribution from the first term will be shown to vanish in the  limit
$s\to
\infty$. Whence we have 
\begin{align*}
  \|P_\c v\|^2=\pi^{-1}\int _0^\infty  \lim_{s\to \infty}\int_{S^{d-1}}
\parb{|u|^2 K_\epsilon m_\epsilon}\parb{\Phi (s,\cdot)}\,\d \sigma\d \lambda.
\end{align*} We are lead to define
\begin{subequations}
\begin{equation}
  \label{eq:diag}
  F^+(\lambda)v=\vGlim_{s\to \infty}\,\pi^{-1/2}\parb{\e^{-\i S_\epsilon} K^{1/2}_\epsilon m^{1/2}_\epsilon R(\lambda +\i 0)v}\parb{\Phi (s,\cdot)},
\end{equation} yielding  
\begin{align*}
  \|P_\c v\|^2=\int _0^\infty  \|F^+(\lambda)v\|_{\vG}^2\,\d \lambda.
\end{align*} Finally the  ``distorted Fourier transform''
\begin{align*}
  F^+:=\int _0^\infty \oplus F^+(\lambda)\,\d \lambda
\end{align*} diagonalizes $H_\c$, i.e. $F^+H_\c=M_\lambda F^+$.

Similarly we can define the ``distorted Fourier transform''
\begin{align*}
  F^-:=\int _0^\infty \oplus F^-(\lambda)\,\d \lambda,
\end{align*} where
\begin{equation}
  \label{eq:diag_-}
  F^-(\lambda)v=\vGlim_{s\to \infty}\,\pi^{-1/2}\parb{\e^{\i S_\epsilon} K^{1/2}_\epsilon m^{1/2}_\epsilon R(\lambda -\i 0)v}\parb{\Phi (s,\cdot)}.
\end{equation}
 \end{subequations}

\subsection{Outgoing approximate generalized
  eigenfunctions} \label{subsec:Approximate generalized eigenfunctions}
 We conclude this section by stating and proving a technical  result
 motivated by the formulas \eqref{eq:diag} and \eqref{eq:diag_-}. This
 enable us to construct outgoing and sufficiently well approximate generalized eigenfunctions
 which in turn are used to
 construct exact  generalized eigenfunctions. 

Let
 $\tau\in C^\infty(S^{d-1})$ and $\lambda\geq 0$ be given. Define a
 function $\tilde u$ by 
 \begin{align}
   \label{eq:tilde u}
   \tilde u=\tilde u(x)= \pi^{1/2}\parb{\chi\e^{\i S_\epsilon} K^{-1/2}_\epsilon
     m^{-1/2}_\epsilon}(x)\tau(\sigma);\;x=\Phi
   (s,\sigma),\;\chi(x)=\chi(|x|). \end{align} Here
 $\chi(r)=\chi(r>2)$ is a cutoff function; see 
 Subsection \ref{subsec:Improved microlocalization for all 
  epsilon-small perturbation} for the precise definition. 
 A short computation (using for example \eqref{eq:gamma_par} stated
 below) shows that
 \begin{align*}
   (H-\lambda)\tilde u=-\pi^{1/2}\chi\e^{\i S_\epsilon(x)} \triangle_x \parb{(K^{-1/2}_\epsilon
     m^{-1/2}_\epsilon)(x)\tau(\sigma)}+\text{ compactly supported term}.
 \end{align*} It will be important for us that also the first term  to
 the right is small at infinity.
 \begin{lemma}
   \label{lemma:diagonalization} Let  
   $\varepsilon>0$ be given and suppose $l=4$. Then for $\epsilon_0>0$ sufficiently
   small, for all $\tau\in C^\infty(S^{d-1})$ and all $\lambda_0>0$
   there  exists $C>0$ such that 
   uniformly in $W_\epsilon$ with 
   $\epsilon\in(0,\epsilon_0]$  and in $\lambda\in [0,\lambda_0]$: 
   \begin{subequations}
   \begin{equation}
     \label{eq:bnd a}
     \forall |\alpha|\leq 2\,\forall |x|\geq 1: \big |K^{1/2}_\epsilon
     m^{1/2}_\epsilon \partial_x^\alpha\parb{(K^{-1/2}_\epsilon
     m^{-1/2}_\epsilon)(x)\tau(\sigma)}\big |\leq C\inp{x}^{\varepsilon-|\alpha|}.
   \end{equation} In particular  the function $\tilde u$ of \eqref{eq:tilde u} obeys
   \begin{equation}
     \label{eq:tilde_ub}
     K^{1/2}_\epsilon
     m^{1/2}_\epsilon(H-\lambda)\tilde u=O\parb{ \inp{x}^{\varepsilon-2}}.
   \end{equation}    
   \end{subequations}
\end{lemma}
 \begin{proof} Let
   \begin{align*}
     T_1(x)&=\int_1^s\parb{K_\epsilon^{-1}\triangle
    S_\epsilon} \parb{\Phi (t,\sigma)}\d t,\\
T_2(x)&=\tau(\sigma);\;x=\Phi(s,\sigma).
   \end{align*} We need to show that
   \begin{equation}
     \label{eq:cla_bnds}
     \forall |\alpha|\leq 2\,\forall |x|\geq 1: \big |\partial_x^\alpha T_j(x)\big |\leq C\inp{x}^{\varepsilon-|\alpha|};\;j=1,2.
   \end{equation} For that we shall use the diffeomorphism $\psi:
   \R^d\to \R^d$ given by 
   \begin{align*}
    y=\Psi (x)= S_0(x)\hat x=\int_0^{|x|}f(r,\lambda)\d r \,|x|^{-1}x,
   \end{align*} and invoke results of \cite{CS} for the model metric
   \begin{align*}
     \tilde g_\epsilon= (\Psi^*)^{-1}g_\epsilon;\;g_\epsilon=K_\epsilon \d x^2.
   \end{align*} This idea of changing framework is actually behind Proposition
   \ref{prop:eikonal-equation} too. Here we shall use  the bounds
   \begin{subequations}
   \begin{align}
     \label{eq:unPs} \forall |\alpha|\leq 2:|\partial^\alpha_x\Psi(x)|&\leq C
     \inp{x}^{-|\alpha|}\inp{\Psi (x)},\\
\label{eq:unPsi-1} \forall |\beta|\leq 2:|\partial^\beta_y\Psi^{-1}(y)|&\leq C
     \inp{y}^{-|\beta|}\inp{\Psi^{-1} (y)},
     \end{align}  which are uniform in $\lambda\in [0,\lambda_0]$.
   \end{subequations}

\noindent {\bf Step I} We note the representation
\begin{equation}
  \label{eq:repGa}
 \Phi (t,\sigma)= \Psi^{-1}\parb{ \tilde \gamma_{\Psi(x)}\parb{t/S_\epsilon(x)}};\;x=\Phi(s,\sigma)=\Phi(S_\epsilon(x),\sigma),
\end{equation} where, using notation of \cite{CS}, $\tilde \gamma_y(t)=ty+\kappa_y(t)$
is  the unique geodesic in  the metric $\tilde
g_\epsilon$ emanating from $0\in\R^d$ with value $y$ at time
one. Whence we can rewrite $T_j(x)$ as 
\begin{align*}
     T_1(x)&=\int_1^{S_\epsilon(x)}\phi\parb{\tilde \gamma_{\Psi(x)}\parb{t/S_\epsilon(x)}}\d t;\;\phi=(K_\epsilon^{-1}\triangle
    S_\epsilon)\circ \Psi^{-1},\\
T_2(x)&=\tau\parb{f(0,\lambda) \Psi^{-1}\parb{\tilde
  \gamma_{\Psi(x)}\parb{1/S_\epsilon(x)}}}. 
   \end{align*} Due to Proposition
   \ref{prop:eikonal-equation}\ref{item:4} and \eqref{eq:unPsi-1} we
   have the bounds (since we have assumed that $l=4$) 
   \begin{equation}
     \label{eq:phi_bnds}
     \forall |\beta|\leq 2:|\partial^\beta_y\phi| \leq C \inp{y}^{-1-|\beta|}.
   \end{equation}

\noindent {\bf Step II} We prove  Sobolev bounds of model geodesics. As in
\cite[Section 6]{CS} introduce  the Sobolev spaces $\vH^p:=W_0^{1,p}(0,1)^d$, $1<p<\infty$,
consisting of absolutely continuous functions $h:[0,1]\to \R^d$
vanishing at the endpoints and having $\dot h\in L^p(0,1)^d=L^p(]0,1[,\R^d)$
(we use the notation $L^p$ for this
vector-valued $L^p$ space). The
space  $\vH^p$ is equipped with the  norm
\begin{equation*}
  \|h\|_{\vH^p}=\|\dot h\|_p=\parbb{\int_0^1|\dot h(t)|^p dt}^{1/p}.
\end{equation*} Due to \cite[Proposition 6.8]{CS}, with reference to
the model geodesic $\tilde\gamma_{y}\parb{t}=ty+\kappa_y(t)$, we have
$\kappa_y\in\vH^p$ for any prescribed  $p\in[2,\infty)$, and for all
sufficiently small $\epsilon>0$
\begin{subequations}
\begin{equation}
      \label{eq:72bB}
      \forall |\beta|\leq 2: \|\partial_y^\beta\kappa\|_{\vH^p}\leq C_p\inp{y}^{1-|\beta|}.
    \end{equation}

We claim that any such fixed $p$ the following generalization holds.
For all $k\in\{0,1,2\}$:
\begin{equation}
      \label{eq:72bBC}
      \forall |\beta|\leq 2:
      \|t^{k-1}\partial_y^\beta\tilde \gamma_y^{(k)}(t)\|_p\leq C_p\inp{y}^{1-|\beta|}.
    \end{equation}  
\end{subequations}  Here $\tilde \gamma_y^{(k)}$ refers to the $k$'th
    time-derivative of  $\gamma=\tilde \gamma_y$. Due to \eqref{eq:72bB} and the Hardy inequality
    \cite[Lemma 6.1]{CS} only the case $k=2$ needs to be proved. But
    since $\gamma$ is a geodesic for the metric $\tilde
    g_\epsilon$  the second derivative $\gamma^{(2)}$ is a sum of
    expressions $\phi_{jk}(\gamma_y)(\dot \gamma_y)^j(\dot
    \gamma_y)^k$  where 
\begin{equation}\label{eq:de2}
     \forall |\beta|\leq 2:|\partial^\beta_z\phi_{jk}| \leq C \inp{z}^{-1-|\beta|}.
   \end{equation} We use the product and chain rules to calculate
   derivatives $\partial_y^\beta$, $|\beta|\leq 2$,  of any such
   expression. Then we can  obtain
   the desired bound  for any term in the resulting expansion by
   combining \eqref{eq:72bBC} for $k=0,1$ (and some bigger values of $p$),
   \eqref{eq:de2}, the a priori bounds
   \begin{equation}\label{eq:bas_ineq}
     ct|y|\leq |\gamma_y(t)|\leq Ct|y|,
   \end{equation} cf.  \cite[Lemma 2.1]{CS}, and the generalized H\"older estimate. We omit the details. The
   reader may consult \cite[Section 6]{CS} for similar arguments.

\noindent {\bf Step III} We can treat $T_1(x)$ by combining
Proposition \ref{prop:eikonal-equation}\ref{item:4},
\eqref{eq:phi_bnds},  \eqref{eq:72bBC} and the generalized H\"older
estimate. The smaller $\varepsilon>0$ is given the bigger $p\geq 2$ in
\eqref{eq:72bBC} is needed. The estimations are straightforward. Let us for completeness do it in details for $|\alpha|=1$:
\begin{align}\label{eq:first_der}
\partial^\alpha
T_1&=(\partial^\alpha S_\epsilon)K_\epsilon^{-1}\triangle
    S_\epsilon+\int_1^{S_\epsilon}\nabla\phi\cdot \parb{(\partial_y\gamma)_{\Psi}\parb{t/S_\epsilon}\cdot
\partial^\alpha \Psi-\dot\gamma_{\Psi}\parb{t/S_\epsilon}\tfrac t{S_\epsilon^2}(\partial^\alpha S_\epsilon)}\,\d t;  
\end{align} The first term is $O\parb{ \inp{x}^{-1}}$. For the second
term we estimate for $\delta=\min(\varepsilon,1)$
\begin{align*}
  |(\nabla\phi)(\gamma_\Psi)|\leq C\big |\tfrac t{S_\epsilon} \Psi\big |^{\delta-2}
\end{align*} and substitute $t\to t{S_\epsilon}$ leading to
  the upper bound of the integral
  \begin{align*}
    Cf\big |\Psi\big |^{\delta-2}\int_0^1
    t^{\,\delta-1}\parb{S_\epsilon  t^{-1}|\partial_y\gamma)_{\Psi}(t)|+|\dot\gamma_{\Psi}(t)|}\,\d t.
  \end{align*} We choose $p\geq 2$ so big that $(\delta-1)/(1-1/p)>-1$
  yielding in turn, using   \eqref{eq:72bBC},  the upper bounds
  \begin{align*}
   C_1(\delta)f S_\epsilon^{\delta-1} \leq C_2(\delta)S_\epsilon^\delta\inp{x}^{-1}=O\parb{ \inp{x}^{\varepsilon-1}}.
  \end{align*}

The case  $|\alpha|=2$ is treated similarly differentiating
\eqref{eq:first_der} except that now there is one
term involving $\dot\gamma_{\Psi}\parb{1}$. For this term we 
 use the formula
\begin{subequations}
\begin{align}\label{eq:la}
  \dot\gamma_{y}\parb{t}=2\int^1_{1/2}\parbb{
    \gamma_y^{(1)}(s)+\int^t_{s}\gamma_y^{(2)}\parb{ t'}\d  t'}\d s
\end{align} with $t=1$ and invoke again \eqref{eq:72bBC}.

\noindent {\bf Step IV} We need to treat $T_2(x)$. In addition
to \eqref{eq:la} we shall use 
\begin{align}\label{eq:lb}
  \gamma_{y}\parb{t}=\int^t_{0}
    \gamma_y^{(1)}(s)\d s.
\end{align} 
 \end{subequations}
The case $|\alpha|=0$ is trivial. We treat the case $|\alpha|=1$  in
details leaving the remaining case $|\alpha|=2$ to the reader (it is
very similar apart from an application of  \eqref{eq:de2} for one term
arising after yet another differentiation):
\begin{align}\label{eq:first_derbb}
\partial^\alpha
T_2&=\nabla\parb{\tau\circ \parb{f(0,\lambda) \Psi^{-1}}}\cdot \partial^\alpha \gamma_{\Psi}\parb{1/S_\epsilon}.  
\end{align} Here the first factor is evaluated in
$\gamma_{\Psi}\parb{1/S_\epsilon}\in S^{d-1}$ and whence,
cf. \eqref{eq:unPsi-1}, it is bounded (uniformly in $\lambda$). For
the second factor of \eqref{eq:first_derbb} we compute
\begin{align}\label{eq:2part}
  \partial^\alpha \gamma_{\Psi}\parb{1/S_\epsilon}=(\partial_y\gamma)_{\Psi}\parb{1/S_\epsilon}\cdot
\partial^\alpha \Psi-\dot\gamma_{\Psi}\parb{1/S_\epsilon}S_\epsilon^{-2}\partial^\alpha S_\epsilon.
\end{align} 

We look at the first term. Using the  H\"older estimate, \eqref{eq:72bBC} and \eqref{eq:lb}
 we estimate
 \begin{align*}
   |(\partial_y\gamma)_{\Psi}(t)|\leq C_p t^{1/p'}; 1/p'+1/p=1,
 \end{align*} which is used with  $t=1/S_\epsilon$. Moreover due to
 \eqref{eq:unPs} we have $|\partial^\alpha \Psi|\leq C
 \inp{x}^{-1}\inp{\Psi}$, so altogether
 \begin{align*}
   |(\partial_y\gamma)_{\Psi}\parb{1/S_\epsilon}\cdot
\partial^\alpha \Psi|\leq C_pS_\epsilon^{1-1/p'}\inp{x}^{-1}; |x|\geq 1.
 \end{align*} If $p\geq 2$ is chosen big enough then 
 $1-1/p'=1/p\leq \varepsilon$, so  the first term of \eqref{eq:2part}
 conforms with \eqref{eq:cla_bnds} with $j=2$ and $|\alpha|=1$.

We look at the second  term. Using the  H\"older estimate, \eqref{eq:72bBC} and \eqref{eq:la}
 we estimate
\begin{align*}
   |\dot\gamma_{\Psi}\parb{t}| \leq C_p t^{-1/p}\inp{\Psi},
 \end{align*}  which again is used with  $t=1/S_\epsilon$, yielding 
\begin{align*}
   |\dot\gamma_{\Psi}\parb{1/S_\epsilon}| \leq C_p S_\epsilon^{1+1/p}; |x|\geq 1.
 \end{align*} Moreover $|S_\epsilon^{-2}\partial^\alpha S_\epsilon|\leq C
 S_\epsilon^{-1}\inp{x}^{-1}$, so altogether
 \begin{align*}
   |\dot\gamma_{\Psi}\parb{1/S_\epsilon}S_\epsilon^{-2}\partial^\alpha
   S_\epsilon|\leq C_p S_\epsilon^{1/p}\inp{x}^{-1},
 \end{align*} which again conforms with \eqref{eq:cla_bnds} with $j=2$ and
 $|\alpha|=1$ provided $p\geq 2$ is chosen as above. 
\end{proof} 
\begin{remark*} The similar result \cite[Proposition 2.5]{ACH} also
     contains a loss of decay (in Lemma \ref{lemma:diagonalization}
     expressed by the power $\inp{x}^{\varepsilon}$). Such loss can in
     general not be avoided. This can seen  using the example in
     Subsection \ref{exa:Example}.
     \end{remark*}

\subsubsection{Generalized eigenfunctions} 
We learn from \eqref{eq:co_area} and  \eqref{eq:tilde_ub}
 that
\begin{equation}
     \label{eq:tilde uc}
     (H-\lambda)\tilde u\in f^{1/2}L^2_{\delta};\;\delta<\tfrac32-\tfrac \mu2-\varepsilon.
   \end{equation} In particular we can choose $\delta>\tfrac 12$ in
   \eqref{eq:tilde uc} provided $\varepsilon>0$ is small enough. With
   such $\delta$ we can define the generalized eigenfunctions
   \begin{equation}
     \label{eq:gen_eig}
     u^\pm =u^\pm(\cdot,\lambda)=\tilde u-R(\lambda \pm\i 0)(H-\lambda)\tilde u.
   \end{equation}
Since intuitively  $u^+$ is a purely outgoing  exact   eigenfunction it
should be zero. This is the content of the following result.
\begin{lemma}
   \label{lemma:diagonalizationb} There exist  $l\geq 4$  and 
  $\epsilon_0>0$  such that for all $\epsilon$-small
  perturbations $W_\epsilon$ with $\epsilon\in(0,\epsilon_0]$ the generalized eigenfunction $u^+$ of \eqref{eq:gen_eig}
   vanishes for any  $\tau\in C^\infty(S^{d-1})$ and any
   $\lambda\geq 0$ .
\end{lemma}\begin{proof}
  By Proposition  \ref{thm:lap} the second term of \eqref{eq:gen_eig}
  is in $f^{-1/2}B(|x|)^*$. The first term is also in this
  space due to an explicit calculation using  
  \eqref{eq:co_area} and the Besov space bound \eqref{eq:besov3} (stated
  below), see \eqref{eq:fund_bound} for a more general statement. So  we conclude  that $f^{1/2}u^\pm \in
  B(|x|)^*$.

Let us
  argue for $\lambda=0$ only.  The case   $\lambda>0$  can be treated
  similarly using  Remark \ref{remark:zero-energy-somm}. To  conclude
  that indeed $u^+=0$ for $\lambda=0$  it
  suffices due to Proposition 
  \ref{thm:somm-radi-condss} to show, with reference to the notation
  \eqref{eq:ab_0}, that for some  small
  positive $\kappa$
  \begin{equation}
    \label{eq:out_expl}
    \Opw(\chi_-(a_0)\tilde \chi_-(b_0))u^+\in (B^
\mu)^*_0\mforall \chi_-\in C^\infty_\c(\R)\mand \tilde \chi_-\in C^\infty_\c((-\infty, \kappa)).
  \end{equation} The contribution from the second term of
  \eqref{eq:gen_eig}, $R(\lambda +\i 0)(H-\lambda)\tilde u$,  is
  treated  by \eqref{eq:31h} (here we may have  $\kappa=1$).  

As for   the
  contribution from the first  term, $\tilde u$,  a computation using  \eqref{eq:bnd
    a} shows  that $f^{-1}(p-\nabla S_\epsilon)\tilde u\in (B^
\mu)^*_0$. On the other hand  due to Proposition
\ref{prop:eikonal-equation}\ref{item:6}  for   small $\kappa,  \epsilon>0$ the symbol
$f^{-1}(\xi-\nabla S_\epsilon)$ is elliptic on the support of any
symbol $\chi_-(a_0)\tilde \chi_-(b_0)$ as in  \eqref{eq:out_expl} which intuitively yields the
desired bound.
 However at  this point some  care must be taken in that
$\partial_jS_\epsilon$ is singular at zero and the good  bounds of Proposition
\ref{prop:eikonal-equation}\ref{item:4} are only valid for $|\alpha|\leq
l$ (which consequently must be chosen sufficiently large). A similar deficiency
 will arise in Section \ref{sec:Quantum bounds}, see \eqref{eq:calc3}
 and the discussion there. Let us give an elaboration: First it is
 more convenient to use $S_0$ (modified by a
 cutoff near infinity)  rather than $S_\epsilon$. Then we have good bounds of all
 derivatives well suited for the calculus of  pseudodifferential
 operators (see Section \ref{sec:Quantum bounds} for some
 details). Whence by
 this  calculus we
 can use the ellipticity property (with the above replacement and  for small  $\kappa>0$) to write, abbreviating  $T=\Opw(\chi_-(a_0)\tilde \chi_-(b_0))$,
 \begin{align*}
   T=T\sum^d_{j=1}T_j \parb{f(r,0)^{-1}p_j-\inp{x}^{-1}x_j}+\widetilde T  \inp{x}^{\mu/2-1},
 \end{align*} where $T_j$ and $\widetilde T $ are bounded
 pseudodifferential operators. We apply this identity
 to $\tilde u$. The last term contributes by a term in $(B^
\mu)^*_0$. As for the first term we use \eqref{eq:bnd
    a} (with $\varepsilon<1-\mu/2$) and get a similar contribution and
  in addition the term
\begin{align*}
   T\sum^d_{j=1}T_j \phi_j\tilde
   u;\;\phi_j=f(r,0)^{-1}\chi(r>1)\partial_jS_\epsilon-\inp{x}^{-1}x_j.
 \end{align*} Next  using   a statement like \eqref{eq:calc3}  
  (see the discussion there) we can write
\begin{align*}
   T\sum^d_{j=1}T_j \phi_j=\sum^d_{j=1}\phi_jT_j T+\widetilde T  \inp{x}^{\mu/2-1}
 \end{align*} for some  $\widetilde T\in \vB\parb{(B^
\mu)^*_0}$. Note  at this point
 that 
 we need   Proposition
\ref{prop:eikonal-equation}\ref{item:4} for an appropriate
$l=l(\mu,d)$. By  Proposition
\ref{prop:eikonal-equation}\ref{item:6}  we then conclude that 
\begin{align*}
   T\tilde u -O(\epsilon^{3/4})T\tilde u \in (B^
\mu)^*_0.
 \end{align*} Consequently (for small  $\epsilon>0$) also $T\tilde u \in (B^
\mu)^*_0$.
 \end{proof} 

\section{Quantum bounds}\label{sec:Quantum bounds}
In this section we collect various microlocal resolvent
  bounds  that will be useful in Section~\ref{sec:Distorted Fourier
    transform} for proving the existence of the limit
  \eqref{eq:diag}  for  $v\in
  L^2_3$ as well as for proving some continuity properties. Our main result Proposition~\ref{Prop:radi-cond-bounds} has
  some independent interest, in particular it is new even for spherically 
  symmetric potentials. 
\subsection{Microlocalization for $\epsilon$-small
  perturbation} \label{subsec:Improved microlocalization for all 
  epsilon-small perturbation}
Let $\tilde r$  denote a smooth convex function of $r\geq 0$ equal to
$1/2$ for $r\leq 1/4$ and   equal to
$r$ for $r\geq 1$.
We introduce for $\lambda\geq 0$ the symbols
\begin{align}\label{eq:ab_lambda}
a_\lambda=a_\lambda(x,\xi)= \frac{\xi^2}{f_\lambda(r)^2},\;\;
b_\lambda=b_\lambda(x,\xi)=  \frac{\xi}{f_\lambda(r)} \cdot
\frac{x}{\tilde r }, 
\end{align} given in terms of  the function  $\tilde r$ of $r=|x|$  and the function $f=f_\lambda=f_\lambda(x)=f_\lambda(r)=f(r,\lambda)$ of Section
\ref{sec:Eikonal equation}. Note that $b_\lambda^2\leq a_\lambda$. We
shall 
state microlocal
properties in terms of these observables. For zero energy  the  
resolvent bounds  of this subsection are
stronger than similar estimates obtainable using the observables
\eqref{eq:ab_0}, see \cite[Proposition 3.5 ii)]{Sk}.
 They are in the
spirit of \cite[Proposition 4.1]{DS3} and \cite[Lemmas 3.2 and
3.3]{Sk}. We shall use 
Weyl quantization of symbols in a uniform symbol class
$S_{\unif} (m_{|z|},g_{|z|})$,
\begin{equation*}
 g=g_{\lambda}=\langle
x \rangle^{-2}\d x^2+f_\lambda(x)^{-2}\d\xi^2. 
\end{equation*} The word \emph {uniform} refers to the
  requirement that bounds of derivatives are uniform in $z$ in the
  closure of 
  $\Gamma_{\theta,\lambda_0}\subset \C$, say $z\in \Gamma^{\clos}_{\theta,\lambda_0}$.  Precisely a symbol $c=c_{z}\in  S_{\unif} (m_{|z|}, g_{|z|})$
 with $z$  in this set, if and only if 
 $c$ obeys the bounds 
 \begin{equation}
   \label{eq:sym_class}
  |\partial^\gamma_x \partial^\beta_\xi c_z(x,\xi)|\leq
  C_{\gamma,\beta}m_{|z|}(x,\xi)\langle x \rangle^{-|\gamma|}f_{|z|}^{-|\beta|}(x).
 \end{equation}
For example $a_{|z|}$ is
  defined for $z\in \Gamma^{\clos}_{\theta,\lambda_0}$ (obviously) and 
  $a_{|z|}\in S_{\unif} (a_{|z|}+1,g_{|z|})$, and similarly for the symbol
  $h_\epsilon$ defined   in \eqref{eq:clasHam} $h_\epsilon,
  h_\epsilon-z\in S_{\unif} (f^2_{|z|}(a_{|z|}+1),g_{|z|})$.  For the
  corresponding calculus the quantity $\inp{x}^{\mu/2-1}$
  plays the role of  a ``uniform Planck constant''. We refer to
  \cite{Sk} for a more elaborate  discussion. The corresponding class of Weyl
  quantized operators is denoted by $\Psi_{\unif} (m_{|z|},g_{|z|})$.

 Consider  real-valued $\chi_-\in C^\infty_\c(\R)$  such that
$\chi_-(t)=1$  in a neighbourhood of $[0,1]$
and
such that $\chi'_-(t)\leq 0$ for $t>0$. Let correspondingly
$\chi_+=1-\chi_-$.  Consider  $\tilde \chi_-\in C^\infty(\R)$ with  $\tilde \chi'_-\in C^\infty_\c((-1,
1))$, $\tilde \chi_-(-1)=1$ and  $\tilde \chi_-(1)=0$. Let $\tilde\chi_+=1-\tilde\chi_-$. Clearly the bound \eqref{eq:limitboundb} below (involving the
function $\chi_+$) is an energy bound. The bound
\eqref{eq:limitboundbb} (involving the
functions $\chi_-$ and $\tilde\chi_-$)  is a microlocal bound whose classical
analogue is partly explained after Proposition \ref{prop:resol_bounds}.

\begin{proposition}
  \label{prop:resol_bounds} Let functions
   $\chi_-$, $\chi_+$  and $\tilde\chi_-$ be given as
  above. There exists
  $\epsilon_0>0$ such the following three properties hold: For all
  $\theta \in (0, \pi)$, $\lambda_0>0$, $\delta>1/2$,
  $t\geq0$ and $\epsilon$-small perturbations $W_\epsilon$ with $\epsilon\in(0,\epsilon_0]$ there
  exists $C>0$ such that
  \begin{enumerate}[i)]
  \item \label{item:7} with $T_+(z):=
\langle x
\rangle^{t-\delta} f_{|z|}^{1/2 }\Opw(a_{|z|}\chi_+(a_{|z|}))R(z)
f_{|z|}^{1/2 }\langle x \rangle^{-t-\delta}$  for  $z \in \Gamma_{\theta,\lambda_0}$
\begin{equation}
  \label{eq:limitboundb}
 \|T_+(z)\| \leq C,
\end{equation}
\item \label{item:8} with  $T_-(z):=\langle x
\rangle^{t-\delta} f_{|z|}^{1/2 }\Opw(\chi_-(a_{|z|})\tilde \chi_-(b_{|z|}))R(z)
f_{|z|}^{1/2 }\langle x \rangle^{-t-\delta}$  for  $z \in \Gamma_{\theta,\lambda_0}$
\begin{equation}
  \label{eq:limitboundbb}
\|
T_-(z)\| \leq C,
\end{equation}
\item \label{item:9}uniformly in $\lambda\in [0,\lambda_0]$ there exist
  \begin{equation}
    \label{eq:limope}
    T_\pm(\lambda+\i0):=\lim_{\Gamma_{\theta,\lambda_0}\ni z \rightarrow \lambda}
T_\pm(z) \text{ in }
{\mathcal B}(L^2).
  \end{equation}
  \end{enumerate}
\end{proposition} 

There are analogous properties for $\bar z\in
\Gamma_{\theta,\lambda_0}$. By the calculus (or the same proof) we can replace the symbol $a_{|z|}\chi_+(a_{|z|})$ in
\ref{item:7} by $\chi_+(a_{|z|}$. In particular the combination of
\ref{item:7} and \ref{item:8} yields an effective microlocalization
$R(z)v\approx \Opw(\chi_-(a_{|z|})\tilde
\chi_+(b_{|z|}))R(z)v$. Let us for later applications choose the localization more
concretely: Let $\chi(\cdot<1)$ be a decreasing smooth function on
$\R$ with $\chi(t<1)=1$ for $t\leq 1/2$ and $\chi(t<1)=0$ for $t\geq
1$. Introduce  for $\kappa>0$ (small) the functions
$\chi(t<\kappa):=\chi(t/\kappa<1)$ and
$\chi(t>\kappa):=1-\chi(t<\kappa)$. Choose
$\chi_-=\chi_\kappa^-=\chi(\cdot-1<\kappa)$ and
$\tilde\chi_+=\tilde\chi_\kappa^+=\chi(1-\cdot<\kappa)$.  This leads
to the introduction of the symbols
\begin{equation}
  \label{eq:loc-func}
  \chi_\kappa=\chi_{\kappa,|z|}=\chi_\kappa^-(a_{|z|})\tilde\chi_\kappa^+(b_{|z|})\in
  S_{\unif} (1,g_{|z|});\;\kappa>0.
\end{equation}

The  proof of Proposition~\ref{prop:resol_bounds} (not to be given in
details 
here) is   similar to the ones of \cite[Lemmas 3.2 and
3.3]{Sk} using instead of \cite[(3.12)]{Sk} the following
 computation, cf. \cite[(4.30)]{DS3}: Let $h_\rad=\xi^2+V_\rad(r)$. The
Poisson bracket with $b=b_\lambda$ (i.e. the derivative of $b$ along the flow
generated by $h_\rad$) is  given by
\begin{subequations}
\begin{equation}
  \label{eq:poisson1}
  \{h_\rad,b\}=\tfrac {2f}{\tilde r}\parb{1-\tfrac
    {rV_\rad'}{2f^2}}\parb{1-b^2}+ \tfrac {2}{f\tilde r}(h_\rad-\lambda).
\end{equation} Due to Condition \ref{cond:rad}\ref{it:assumption3r}
the factor $1-\tfrac
    {rV_\rad'}{2f^2}\geq \tilde \epsilon_1/2$. For the
    Hamiltonian $h_\epsilon$ of \eqref{eq:clasHam} we have uniformly 
    $x,\xi\in \R^d$ and $\lambda\geq 0$
\begin{equation}
  \label{eq:poisson1b}
  \{h_\epsilon,b\}=\tfrac {2f}{\tilde r} \parb{1-\tfrac
    {rV_\rad'}{2f^2}}\parb{1-b^2}+ \tfrac {2}{f\tilde
    r}(h_\epsilon-\lambda)+O(\epsilon )\tfrac {f}{\tilde r}.
\end{equation}  Whence 
uniformly in a  set of the form $\{b^2\leq 1-\delta\}$, $\delta>0$,  and
    $\lambda\geq 0$
\begin{align}
  \begin{split}
  \{h_\epsilon,b\}&\geq\tilde \epsilon_1\tfrac {f}{\tilde r} \parb{1-b^2}+ \tfrac {2}{f\tilde
    r}(h_\epsilon-\lambda)-\epsilon C\tfrac {f}{\tilde r}\\
  &\geq\tfrac {f}{\tilde r} \parb{\tilde \epsilon_1\delta- C\epsilon}+ \tfrac {2}{f\tilde
    r}(h_\epsilon-\lambda).\label{eq:poisson2}
  \end{split}
\end{align}
 \end{subequations} We learn from \eqref{eq:poisson2} that  provided
 $\epsilon$ is taken small the
 observable $b$ grows along the flow
generated by $h_\epsilon$ on any set $\{b^2\leq
 1-\delta,\;h_\epsilon=\lambda\}$. This is part of the classical analogue of
 \eqref{eq:limitboundbb}. For $\kappa$-depending symbols used as input in Proposition
 \ref{prop:resol_bounds} the bounds \eqref{eq:poisson2} indicate an 
 optimal choice, $\epsilon_0\approx \kappa$. In fact, and more
 precisely, the proof of Proposition
 \ref{prop:resol_bounds}  shows that we can choose
 $\epsilon_0=\kappa/C$ for some $C>0$ in the  regime
 $\kappa>0$ small, whence allowing us to write $R(z)v\approx
 \Opw(\chi_{\kappa,|z|})R(z)v$ for all $(\kappa/C)$-small perturbations.
 This will be one reason for  considering perturbations of a spherically 
 symmetric potential only. Another reason  originates  in the construction of
 $S_\epsilon$, i.e.  Proposition \ref{prop:eikonal-equation}.

\subsection{Preliminary considerations} \label{subsec:Preliminary considerations}
We assume in this subsection that 
$V_2=0$. Whence we consider here the quantization of
\eqref{eq:clasHam}, $H=H_\epsilon$.
\subsubsection{Calculus considerations} \label{subsubsec:Calculus
  considerations}  By the calculus the family of symbols  \eqref{eq:loc-func}
has the properties  that for all $n\in \N$ and all $c=c_z\in S_{\unif}
(a_{|z|}+1,g_{|z|})$
\begin{subequations}
\begin{align}\label{eq:calc1}
  \parb{I-\Opw(\chi_{2\kappa,|z|})}\Opw(\chi_{\kappa,|z|})&\in
    \Psi_{\unif} (\inp{x}^{-n}\inp{\xi}^{-2},g_{|z|}),\\
[\Opw(c_z),\Opw(\chi_{\kappa,|z|})]\Opw(\chi_{\kappa/2,|z|})&\in \Psi_{\unif}
(\inp{x}^{-n}\inp{\xi}^{-2},g_{|z|})\label{eq:calc2}.
\end{align}
 Note that in  particular \eqref{eq:calc2} applies
  to $c_z=h-z$ and any function $c_z=c_z(x,\xi)=\phi_z(x)\in S_{\unif}
(1,g_{|z|})$. We shall need the following modification of
the latter statement. Consider  $\phi_z\in C^N(\R^d)$, with  $z\in
\Gamma^{\clos}_{\theta,\lambda_0}$, obeying  uniform bounds
  \begin{align*}
      |\partial_x^\alpha\phi_z(x)|\leq C_\alpha f_{|z|}(x)\langle x\rangle^{-|\alpha|}  ;\; |\alpha|\leq N.
    \end{align*} Now for any given $n\in \N$ we can find $N=N(n, \mu,d)$
    such that  for all  $\phi_z\in C^N(\R^d)$ obeying these bounds we have (uniformly in $z\in
\Gamma^{\clos}_{\theta,\lambda_0}$)
    \begin{align}\label{eq:calc3}
      \begin{split}
        [\phi_z(x),\Opw(\chi_{2\kappa,|z|})]\Opw(\chi_{\kappa,|z|})&=\inp{x}^{-n}B\inp{x}^{-n};\;
\|B\|\leq C.
      \end{split}
    \end{align}
 \end{subequations} This statement can be proved by the symbolic
 calculus and an explicit estimation of an associated oscillatory
 integral. Note that the constant $C$ in \eqref{eq:calc3}  can be
 chosen proportional to a  natural norm of $\phi_z$, and whence
 the bound is an example  of a familiar continuity property of the calculus of
 pseudodifferential operators.  We shall apply it to
 $\phi_z=\chi(r>2)\partial_jS_\epsilon(x,|z|)=\chi(r>2)\partial_jS_\epsilon(x,|z|)$, $j=1,\dots,d$.
 Note that we here need $l=N+1$ in Proposition
 \ref{prop:eikonal-equation}\ref{item:4}. If $n$ in
 \eqref{eq:calc3} is taken 
 large possibly (in fact likely
 so) $\epsilon_0>0$ in Proposition
 \ref{prop:eikonal-equation} must then be  small (since in practise   $N=l-1$
 large is  needed for \eqref{eq:calc3} for  given large $n$).

The last preliminary property we will discuss is an application of the
Fefferman-Phong inequality \cite[Theorem 18.6.8]{Ho2} (uniform
version), concretely  bounds for the symbol $b_{|z|}$ and $\chi_{2\kappa,|z|}$:
For all $\kappa>0$ 
there exists $C=C_\kappa>0$ so that for all $z\in
\Gamma^{\clos}_{\theta,\lambda_0}$
\begin{subequations}
\begin{align}
  \label{eq:Fef_Phong1}
  \Opw(\chi_{2\kappa,|z|})\Opw(b_{|z|}-1+2\kappa)\Opw(\chi_{2\kappa,|z|})&\geq
  -C(\inp{x}f_{|z|})^{-2},\\
 \Opw(\chi_{2\kappa,|z|})\Opw(b_{|z|}-1-2\kappa)\Opw(\chi_{2\kappa,|z|})&\leq
  C(\inp{x}f_{|z|})^{-2},\label{eq:Fef_Phong2}\\
\Opw(\chi_{2\kappa,|z|})\Opw\parb{(b_{|z|}-1)^2-(2\kappa)^2}\Opw(\chi_{2\kappa,|z|})&\leq
  C(\inp{x}f_{|z|})^{-2}.\label{eq:Fef_Phong3}
\end{align} 
\end{subequations}

\subsubsection{Radiation operators} \label{subsubsec:Radiation operators}
We shall combine Propositions  \ref{prop:eikonal-equation} and 
\ref{prop:resol_bounds} to obtain radiation condition bounds similar
to some  of \cite{HS1,HS2} for positive energies (see also \cite{Sa1,Sa2}). Our method is different in that it is
purely stationary whereas \cite{HS1,HS2} rely on propagation estimates. Whence we introduce for $\lambda \geq 0$ and any given
$\epsilon$-small perturbation $W_\epsilon$  \emph {radiation
  operators} defined  in terms of the function
$S_\epsilon=S_\epsilon(x,\lambda)$ from Proposition
\ref{prop:eikonal-equation} as
\begin{align*}
  \gamma=p-\nabla
  S_\epsilon,\;\gamma_j=p_j-\partial_jS_\epsilon,\,j=1,\dots, d,\mand \gamma_\|=\Re \parb{ \nabla
    S_\epsilon\cdot \gamma}.
\end{align*} Using \eqref{eq:eikonal} we obtain, cf. \cite{HS1,HS2},
\begin{equation}
  \label{eq:gamma_par}
  2\gamma_\|=(H-\lambda)-\gamma^2.
\end{equation}

Next we compute the Heisenberg derivative, say
denoted by $\D=\i[H,\cdot]$,  of $\gamma$.  The 
involved operators are local and we shall only need the computation for
$r\geq 1$.
\begin{align}\label{eq:GG}
  \begin{split}
 \D\gamma&=-2\nabla^2S_\epsilon \gamma+\i2 \nabla \triangle
 S_\epsilon\\
&=-\tfrac {2f}{ r}\parb{\gamma-\parb{1+\tfrac
    {rV_\rad'}{2f^2}}f^{-1}|\hat x\rangle \nabla
  S_\epsilon\cdot\gamma+O(\epsilon^{1/2})\gamma}+\i2 \nabla
\triangle S_\epsilon\\
&=-\tfrac {2f}{ r}\parb{\gamma+F f^{-1} \parb{2\gamma_\|+\i\triangle
  S_\epsilon}+O(\epsilon^{1/2})\gamma}+\i2 \nabla
\triangle S_\epsilon\\
&=-\tfrac {2f}{ r}\parb{\gamma+F f^{-1} \parb{(H-\lambda)-\gamma^2}+O(\epsilon^{1/2})\gamma}-2\i\parb{\tfrac {F }{ r}\triangle
  S_\epsilon- \nabla
\triangle S_\epsilon};\\
&
\,\,\,\,\,\,\,\,\,\,\,\,\,\,\,\,\,\,\,\,\,\,\,\,\,\,\,\,\,\,\,\,\,\,\,\,\,\,\,\,F:=\tfrac{\tilde
f}{\tilde r} x,\; \tilde f=-\tfrac{1}{2}\parb{1+\tfrac
    {rV_\rad'}{2f^2}}.
    \end{split}
\end{align} Here we used Proposition
\ref{prop:eikonal-equation}\ref{item:6} and \eqref{eq:gamma_par}. The meaning
of $O(\epsilon^{1/2})$ is the same as in the proposition, i.e. it is a
uniform bound. We can
simplify the right hand side using Proposition
\ref{prop:eikonal-equation}\ref{item:4} (assuming $l\geq 3$) and conclude that 
\begin{align}\label{eq:Dgamma}
 \D \gamma=-\tfrac {2f}{
   r}\parbb{\parb{I+O(\epsilon^{1/2})}\gamma+F f^{-1} (H-\lambda)-F f^{-1} \gamma^2-\i
 r^{-1}O(\epsilon^{0})},
\end{align} where as above the estimates are uniform in $W_\epsilon$, 
   $\lambda \geq 0$ and $x$ with $r=|x|\geq 1$.

Next we compute
\begin{align}\label{eq:conGammaSquare}
  \Re\parb{f^{-1}F\cdot \gamma}= \Re\parb{\tilde f \parb{\Opw(b)-1+O(\epsilon^{3/4})}}.
\end{align} Effectively the right hand side will be ``small''; we
will use it to treat the third term in \eqref{eq:Dgamma}. Here it is
also useful to note that
\begin{equation}
  \label{eq:Gamma_comm}
  [\gamma_i,\gamma_j]=0 \for 1\leq i,j\leq d.
\end{equation}

\subsection{Strong radiation condition bounds} \label{subsec:Radiation condition bounds}
We introduce for $k\in \N$
\begin{equation*}
  X=\inp{x}=(1+r^2)^{1/2}\mand X_k=X(1+r^2/k)^{-1/2},
\end{equation*} and ``propagation observables'' 
\begin{subequations}
\begin{align}\label{eq:Q1}
  &P_1=\sum_{i}Q_i^*Q_i;\; &&\;Q_i=X_k^{1-\epsilon'}\chi(r)\gamma_i\Opw(\chi_{\kappa,|z|}),\\
&P_2=\sum_{i,j}Q_{ij}^*Q_{ij};\; &&Q_{ij}=X_k^{2(1-\epsilon')}\chi(r)\gamma_i\gamma_j\Opw(\chi_{\kappa,|z|}),\label{eq:Q2}
\end{align}  
\end{subequations} where  $\gamma=\gamma(\lambda=|z|)$, $\chi(r)=\chi(r>2)$ and
$\epsilon'\in (0,1]$ needs to be specified. Note that the
powers of  $X_k$ are bounded factors and that pointwise $X_k\uparrow X$  for $k\to \infty$. 

We compute the Poisson brackets
\begin{subequations}
\begin{align}
\label{eq:Gamma_comm0}
  \{h,\chi(r)\}=2 \chi'(r)\hat x\cdot \xi &=2f\chi'(r) \,b,\\
  \label{eq:Gamma_comm1}
  \{h,X\}=2X^{-1} x\cdot \xi &=\tfrac {2f}r \phi \,b X;\;&&\,\;\phi=\tfrac{r \tilde
    r}{X^2},\\
\label{eq:Gamma_comm3}
  \{h,X_k\}&=\tfrac {2f}r (\phi-\phi_k) \,b X_k;\;&&\phi_k=\tfrac{r \tilde r}{k+r^2}.
\end{align} 
\end{subequations} Note for \eqref{eq:Gamma_comm3} that
\begin{equation}
  \label{eq:reg-bound}
  0\leq
\phi-\phi_k=\tfrac{(k-1)r \tilde r}{(k+r^2)X^2}\leq 1.
\end{equation}
  
Yet another property we will use (tacitly)  are the uniform bounds 
\begin{align*}
  |\partial_x^{\alpha}f_\lambda^s|&\leq C_{\alpha,s}f_\lambda^s\langle x
\rangle^{-|\alpha|};\;\lambda\geq 0,\; x\in\R^d,\\
|\partial_x^{\alpha}X_k^s|&\leq C_{\alpha,s}X_k^s\langle x
\rangle^{-|\alpha|};\;k\in \N,\; x\in\R^d.
\end{align*}

The main result of this section is
\begin{proposition}
  \label{Prop:radi-cond-bounds} There exist  $l=l(\mu,d)\in \N$ ($l\geq
  4$ is used explicitly) and 
  $\epsilon_0,C_0>0$ with  $\sqrt C_0\epsilon_0\leq 1$,  such that for all $\epsilon$-small
  perturbations $W_\epsilon$ with $\epsilon\in(0,\epsilon_0]$
   and with
  $\epsilon'= C_0\sqrt \epsilon$ the following bounds \eqref{eq:Q_1_bound}
  and \eqref{eq:Q_2_bound} hold uniformly in
  $\lambda$ in intervals
  of the form $I=[0,\lambda_0]$. We consider in these bounds
  components $\gamma_i,\gamma_j$, $1\leq i,j\leq d$,
of 
  $\gamma=\gamma_\epsilon(\lambda)=p-\nabla S_\epsilon(x,\lambda)$. 
  \begin{subequations}
 \begin{align}
    \label{eq:Q_1_bound}
    \begin{split}
      \|X^{1-\epsilon'}&\parb {\tfrac{f_\lambda}{r}}^{1/2}\chi(r)\gamma_iR(\lambda+\i
    0)f_\lambda^{1/2}X^{-3/2}\|\leq C,
    \end{split}
\end{align} 
\begin{align}
     \label{eq:Q_2_bound}
\begin{split}\|X^{2(1-\epsilon')}&\parb {\tfrac{f_\lambda}{r}}^{1/2}\chi(r)\gamma_i\gamma_jR(\lambda+\i
    0)f_\lambda^{1/2}X^{-5/2}\|\leq C.
 \end{split}
  \end{align} 
  \end{subequations}
\end{proposition}
\begin{proof} Due to  resolvent equations  we can assume that 
$V_2=0$. Throughout the proof  the notations  $H$ and $R(z)$ refer to this case.
  Fix $\theta \in (\pi/2, \pi)$ and  $\lambda_0>0$. We shall prove
  microlocal bounds of states $u=R(z)v$ in terms of 
   quantities related to $P_1$ of \eqref{eq:Q1} and $P_2$ of
   \eqref{eq:Q2}, respectively, where   $z\in
  \Gamma_{\theta,\lambda_0}$. In particular we consider below $\gamma_i=p_i-\partial_i
  S_\epsilon(x,|z|)$  with $z\in
  \Gamma_{\theta,\lambda_0}$. We could choose to 
  take $\kappa>0$ in the definition of the factors
  $\Opw(\chi_{\kappa,|z|})$ to be proportional to $\epsilon$ with   
a sufficiently large constant of proportionality, cf. a discussion at the end of Subsection
  \ref{subsec:Improved microlocalization for all epsilon-small
    perturbation}. However the larger choice $\kappa=\epsilon^{3/4}$
  suffices and will be used below. In any case for the corresponding
  lozalization operators $B=\Opw(\chi_{\kappa,|z|})$ and  $B=\Opw(\chi_{\kappa/2,|z|})$ we can use the
  bounds of Proposition \ref{prop:resol_bounds} for $\epsilon$-small
  perturbations (more precisely we have such  bounds upon replacing
  the pseudodifferential operators there by $I-B$). This is done in   \eqref{eq:q1bndb} and
  \eqref{eq:second2b} below (for \eqref{eq:Q_1_bound}). We choose $\epsilon_0>0$ in agreement
  with any such 
  application as well as being in agreement with Proposition
  \ref{prop:eikonal-equation} with  an $l$ (the one to be used in the
  proposition)  choosen sufficiently
  large. How large $l$ must be depends for \eqref{eq:Q_1_bound} partly on applications below of
  Proposition
  \ref{prop:eikonal-equation} and the symbolic calculus property
  \eqref{eq:calc3}, used in  \eqref{eq:3_cont} and \eqref{eq:S3}. See
  \eqref{eq:calc3bb} 
  for the case of \eqref{eq:Q_2_bound} (used in \eqref{eq:3_cont2} and
  \eqref{eq:S3bb}). Of course it is
  legitimate to take $\epsilon_0$ smaller if needed. The choice $\epsilon'= C_0\sqrt
  \epsilon$ for some (large) constant $C_0$ (rather than  $\epsilon'$
  being proportional to $\epsilon$) is needed (and best possible) in our
  treatment of the contribution from the term
  $O(\epsilon^{1/2})\gamma $ in \eqref{eq:Dgamma} in the  computation
  and estimation of a commutator, see \eqref{eq:main_cont}  below. We fix an applicable 
  $C_0$ for \eqref{eq:Q_1_bound}  at the end of Step I (this constant will also
  work for \eqref{eq:Q_2_bound}, see  the end of Step II). 
 
\noindent {\bf Step I} We show \eqref{eq:Q_1_bound} by first
establishing the bound
\begin{align}\label{eq:Q1bnd}
\inp{P'_1}_u&\leq C_{1}
  \|f_{|z|}^{-1/2}X^{3/2}v\|^2+C_{2}\tfrac {|(z-|z|)| ^2}{\Im z}\|X^{{-1}}X_k^{{1-\epsilon'}}u\|^2;\\P'_1&=\sum_iQ_i^*\tfrac{f_{|z|}}{r}Q_i.\nonumber
\end{align} Here we have suppressed the dependence of $|z|$ in $Q_i$
(as above). The constants  are independent of
$z\in \Gamma_{\theta,\lambda_0}$ and $k\in \N$ (however dependent on
$\epsilon$ and possibly also $W_\epsilon$). Whence we  conclude
by first letting $\Im z\to 0$ (for fixed $\lambda=\Re z
\geq 0$) and then letting
$k\to \infty$, that at all energies $\lambda\in [0,\lambda_0]$
\begin{align*}
  \sum_i \inp{Q_i^*\tfrac{f_{\lambda}}{r}Q_i}_{u}\leq C_{1}
  \|f_{\lambda}^{-1/2} X^{3/2}v\|^2,
\end{align*} and whence
\begin{subequations}
\begin{align}
  \label{eq:q1bnda}
  \|X^{1-\epsilon'}\parb {\tfrac{f_\lambda}{r}}^{1/2}\chi(r)\gamma_i\Opw(\chi_{\kappa,\lambda})R(\lambda+\i
    0)v\|\leq \sqrt{C_{1}}\|f_{\lambda}^{-1/2}X^{3/2}v\|.
\end{align}
On the other hand we have, cf. Proposition \ref{prop:resol_bounds},  
\begin{align}
  \label{eq:q1bndb}
  \|X^{1-\epsilon'}\parb {\tfrac{f_\lambda}{r}}^{1/2}\chi(r)\gamma_i\parb{I-\Opw(\chi_{\kappa,\lambda})}R(\lambda+\i
    0)v\|\leq C_3\|f_{\lambda}^{-1/2}X^{3/2}v\|.
\end{align}
 \end{subequations} Clearly \eqref{eq:Q_1_bound} follows from
 \eqref{eq:q1bnda} and \eqref{eq:q1bndb}.

To show \eqref{eq:Q1bnd} we calculate the expectation
\begin{subequations}
\begin{align}
  \label{eq:com1}
  \inp{\i [H,P_1]}_u&=2\Im z\,\inp{P_1}_u+2 \Im \inp{P_1u,v},\\
\label{eq:com2}
  \inp{\i [H,P_1]}_u&=2\sum_i\Re \inp{Q_iu,\i [H,Q_i]u},\\
\label{eq:com3}
  \i [H,Q_i]&=T_i^1+T_i^2;\; \\
&T_i^1=\i
[H,X_k^{1-\epsilon'}\chi(r)\gamma_i]\Opw(\chi_{\kappa,|z|}),\nonumber\\
&T_i^2=X_k^{1-\epsilon'}\chi(r)\gamma_i\i [H,\Opw(\chi_{\kappa,|z|})].\nonumber
\end{align}
 \end{subequations}  The idea of the proof  is to show that \eqref{eq:com1} ``tends''
 to be non-negative while \eqref{eq:com2} ``tends''
 to be non-positive. To keep the notation at a  minimum we abbreviate
 $f_{|z|}=f$ 
 in the remaining part of the proof of the proposition.

Clearly indeed the  first term to the right in \eqref{eq:com1} is
non-negative (to be used in \eqref{eq:2_cont} stated below).

The 
second  term to the right in \eqref{eq:com1} can be estimated as
\begin{align}
  \label{eq:first}
  \begin{split}
  |2 \Im \inp{P_1u,v}|&\leq \delta\inp{P'_1}_{
    u}+\delta^{-1}\sum_i \inp{r/f}_{Q_i
    v}\\
&\leq \delta\inp{P'_1}_{
    u}+\delta^{-1}C_1\|X^{{3/2-\epsilon'}}f^{1/2}v\|^2\\
&\leq \delta\inp{P'_1}_{
    u}+\delta^{-1}C_2\|f^{-1/2}X^{m}v\|^2;\;m\geq 3/2-\epsilon'.
  \end{split}
\end{align} We choose $\delta>0$ suitably small later.

As for \eqref{eq:com2} we substitute \eqref{eq:com3}. The contribution
from the terms $T_i^2$ is  estimated similarly using
\eqref{eq:calc2} (for suitable $n$) and  
\eqref{eq:limitbound2j}
\begin{subequations}
\begin{align}
  \label{eq:second2}
  \begin{split}
  &2\sum_i|\inp{Q_iu,T_i^2\Opw(\chi_{\kappa/2,|z|})u}|\\ &\leq
  \tfrac\delta 2 \inp{P'_1}_{
    u}+\delta^{-1}C_1\|f^{-1/2}X^{m}v\|^2;\;m>1/2,
  \end{split} 
\end{align} and by using  Proposition \ref{prop:resol_bounds}
\begin{align}
  \label{eq:second2b}
  \begin{split}
  &2\sum_i|\inp{Q_iu,T_i^2\parb{I-\Opw(\chi_{\kappa/2,|z|})}u}|\\ &\leq \tfrac\delta 2 \inp{P'_1}_{
    u}+\delta^{-1}C_2\|f^{-1/2}X^{m}v\|^2;\;m>3/2-\epsilon'.
  \end{split} 
\end{align}   
\end{subequations} 

It remains to consider the contribution from $T_i^1$. We split
\begin{align*}
  T_i^1&=S_i^1+S_i^2+S_i^3;\\
&S_i^1=X_k^{1-\epsilon'}\chi(r)\parb{\D \gamma_i}\Opw(\chi_{\kappa,|z|}),\\
&S_i^2=X_k^{1-\epsilon'}\parb{\D
  \chi(r)}\gamma_i\Opw(\chi_{\kappa,|z|}),\\
&S_i^3=\parb{\D
  X_k^{1-\epsilon'}}\chi(r)\gamma_i\Opw(\chi_{\kappa,|z|}),
\end{align*} and intend to use \eqref{eq:Dgamma},
\eqref{eq:Gamma_comm0} and \eqref{eq:Gamma_comm3} to treat the
contribution to \eqref{eq:com2} from the three
terms, respectively.

The seeked negativity comes from the terms $S_i^1$ more precisely from the
contribution from the first term to the right in
\eqref{eq:Dgamma}. Thus, by the Cauchy Schwarz inequality, 
\begin{align}\label{eq:main_cont}
  \begin{split}
  &2\sum_i\Re \inp{Q_iu,X_k^{1-\epsilon'}\chi(r)\tfrac {-2f}{
   r}\parb{\gamma_i+\parb{O(\epsilon^{1/2})\gamma}_i}\Opw(\chi_{\kappa,|z|})u}\\
&\leq \parb{-4+\tilde C\sqrt \epsilon}\inp{P'_1}_u.
\end{split}
\end{align}

To bound the contribution from the second term to the right in
\eqref{eq:Dgamma} we use that $F$ is (uniformly) bounded and estimate
\begin{align}\label{eq:2_cont}
  \begin{split}
  &2\sum_i\Re \inp{Q_iu,X_k^{1-\epsilon'}\chi(r)\tfrac {-2F^i}{
   r}(H-|z|)\Opw(\chi_{\kappa,|z|})u}\\
&\leq -4(z-|z|)\sum_i\Re \inp{Q_iu,X_k^{1-\epsilon'}\chi(r)\tfrac {F^i}{
   r}\Opw(\chi_{\kappa,|z|})u} \\ & \quad + \delta\inp{P'_1}_{
    u}+\delta^{-1}C_1\|f^{-1/2}X^{m}v\|^2\\
&\leq 2\Im z\,\inp{P_1}_u+\tilde C_2\tfrac {|(z-|z|)| ^2}{\Im
  z}\|X^{{-1}}X_k^{{1-\epsilon'}}u\|^2\\
\\ & \quad + \delta\inp{P'_1}_{
    u}+\delta^{-1}C_1\|f^{-1/2}X^{m}v\|^2;\;m>1/2.
\end{split}
\end{align}

To bound the contribution from the third term to the right in
\eqref{eq:Dgamma} we 
``redistribute'' the factors of components in $\gamma^2$  and use
\eqref{eq:limitbound2j}, 
\eqref{eq:calc1}, \eqref{eq:calc3}, \eqref{eq:Fef_Phong3},
\eqref{eq:conGammaSquare} and
\eqref{eq:Gamma_comm}  estimating with  $u_j:=\Opw(\chi_{2\kappa,|z|})(2f/
   r)^{1/2}Q_ju$ 
\begin{align}\label{eq:3_cont}
  \begin{split}
  &2\sum_i\Re \inp{Q_iu,X_k^{1-\epsilon'}\chi(r)\tfrac {2F^i}{
   r}\gamma^2\Opw(\chi_{\kappa,|z|})u}\\
&\leq 2\sum_j\Re \inp{\tfrac {2}{
   r}F\cdot \gamma Q_ju, Q_ju}+C_1\|f^{-1/2}X^{m}v\|^2\\
&=2\sum_j\Re \inp{\tilde
  f \parb{\Opw(b)-1 +O(\epsilon^{3/4})} }_{ u_j}+C_2\|f^{-1/2}X^{m}v\|^2\\
&\leq \parb{2\kappa (\sup |\tilde f|^2+1)+C_1\epsilon^{3/4}}\inp{2P'_1}_u +C_3\|f^{-1/2}X^{m}v\|^2\\ & 
=  C_4\epsilon^{3/4}\inp{P'_1}_u +C_3\|f^{-1/2}X^{m}v\|^2;\;m>1/2.
\end{split}
\end{align} 

To bound the contribution from the fourth term to the right in
\eqref{eq:Dgamma} it is convenient (although not necessary) to
symmetrize (assuming then $l\geq 4$). This gives with $\tilde u:=X_k^{1-\epsilon'}\Opw(\chi_{\kappa,|z|})u$
\begin{align}\label{eq:4_cont}
  \begin{split}
 \MoveEqLeft 2\sum_i\Re \inp{Q_iu,X_k^{1-\epsilon'}\chi(r)
\i\tfrac {2f}{
   r^2}O_i(\epsilon^{0})\Opw(\chi_{\kappa,|z|})u}\\
&\leq C_1\inp{
\tfrac {2f}{
   \tilde r^3}}_{\tilde u}\\
&\leq C_2\|f^{-1/2}X^{m}v\|^2;\;m>1/2.
\end{split}
\end{align}
  
Clearly the contribution from from the terms $S_i^2$ are bounded
similarly, cf. \eqref{eq:Gamma_comm0},
\begin{align}\label{eq:S2_cont}
  2\sum_i\Re \inp{Q_iu,S_i^2u}
\leq C\|f^{-1/2}X^{m}v\|^2;\;m>1/2.
\end{align}
 
It remains to examine the contribution from from the terms $S_i^3$. We
shall use \eqref{eq:calc1}, \eqref{eq:calc3}, \eqref{eq:Fef_Phong2}, \eqref{eq:Gamma_comm3} and
\eqref{eq:reg-bound} estimating with 
$u_i:=\Opw(\chi_{2\kappa,|z|})((\phi-\phi_k)2f/r)^{1/2}Q_iu$,
\begin{align}\label{eq:S3}
  \begin{split}
  \MoveEqLeft 2\sum_i\Re \inp{Q_iu,S_i^3\Opw(\chi_{\kappa,|z|})u}\\
&\leq 2(1-\epsilon')\sum_i \inp{\Opw(b)}_{u_i}+C_1\|f^{-1/2}X^{m}v\|^2\\
&\leq 2(1-\epsilon')(1+2\kappa)\sum_i\|
u_i\|^2+C_2\|f^{-1/2}X^{m}v\|^2\\
&\leq 4(1-\epsilon')(1+2\epsilon^{3/4} )\inp{P'_1}_u +C_3\|f^{-1/2}X^{m}v\|^2;\;m>1/2.
\end{split}
\end{align} 

Now  by combining \eqref{eq:first}--\eqref{eq:S3} with
\eqref{eq:com1}--\eqref{eq:com3}  we obtain
\begin{align}\label{eq:concl}
  \begin{split}
  &\parb{4-\tilde C \sqrt \epsilon-3\delta - C_4\epsilon^{3/4}
   -4(1-\epsilon')(1+2\epsilon^{3/4})}\inp{P'_1}_u\\
 &\leq C_\delta\|f^{-1/2}X^{m}v\|^2+\tilde C_2\tfrac {|(z-|z|)| ^2}{\Im
  z}\|X^{{-1}}X_k^{{1-\epsilon'}}u\|^2;\;m> 3/2-\epsilon'.
 \end{split}
\end{align}
We choose $\delta=\sqrt \epsilon$ and fix $C_0= \tfrac14(\tilde
C+5)$. Then (possibly by taking $\epsilon_0>0$ smaller) we conclude  the bound 
\begin{align*}
  \sqrt \epsilon\inp{P'_1}_u\leq C_\delta\|f^{-1/2}X^{m}v\|^2+\tilde C_2\tfrac {|(z-|z|)| ^2}{\Im
  z}\|X^{{-1}}X_k^{{1-\epsilon'}}u\|^2;\;m=3/2,
\end{align*} whence  \eqref{eq:Q1bnd} follows.

\noindent {\bf Step II} We show \eqref{eq:Q_2_bound} by 
establishing the bound
\begin{align}\label{eq:Q1bnd2}
\inp{P'_2}_u&\leq C_{1}
  \|f^{-1/2}X^{5/2}v\|^2+C_{2}\tfrac {|(z-|z|)| ^2}{\Im z}\parb{\|X^{{-2}}X_k^{{2(1-\epsilon')}}u\|^2+k^2\inp{P'_1}_u}+C_{3}\inp{P'_1}_u;\\P'_2&=\sum_iQ_{ij}^*\tfrac{f}{r}Q_{ij}.\nonumber
\end{align} Here $P'_1$ is given as in \eqref{eq:Q1bnd}. Due to
\eqref{eq:Q1bnd} we can proceed as above letting first  $\Im z\to 0$ (for fixed $\lambda=\Re z
\geq 0$) and then 
$k\to \infty$. Then again we invoke Proposition
\ref{prop:resol_bounds}. Whence it suffices for \eqref{eq:Q_2_bound}
to show \eqref{eq:Q1bnd2}.

For \eqref{eq:Q1bnd2} we proceed similarly as in Step I giving now
less details.
We replace $P_1$ by $P_2$ in \eqref{eq:com1}--\eqref{eq:com3}
and need to show ``essential positivity'' and  ``essential negativity'' of the expression to the right
of the analogous \eqref{eq:com1} and \eqref{eq:com2}, respectively.
The most interesting contribution to the analogous commutator \eqref{eq:com2} is
 the one from an expression like $S^1_i$. More precisely this term  is  now 
replaced by 
\begin{align*}
  S_{ij}^1=X_k^{2(1-\epsilon')}\chi(r)\parb{\D (\gamma_i\gamma_j)}\Opw(\chi_{\kappa,|z|}).
\end{align*} We write
\begin{align*}
  \D (\gamma_i\gamma_j)=(\D \gamma_i)\gamma_j+\gamma_i(\D \gamma_j),
\end{align*} and invoke again \eqref{eq:Dgamma} which contains four terms.

As for the first term the analogous of  \eqref{eq:main_cont} reads
(using the constant $\tilde C$ from \eqref{eq:main_cont})
\begin{align}\label{eq:main_contII}
  \begin{split}
  &2\sum_{i,j}\Re \inp{Q_{ij}u,X_k^{2(1-\epsilon')}\chi(r)\\ & \quad \quad \parbb{\tfrac {-2f}{
   r}\parb{\gamma_i+\parn{O(\epsilon^{1/2})\gamma}_i}\gamma_j+\gamma_i\tfrac {-2f}{
   r}\parb{\gamma_j+\parn{O(\epsilon^{1/2})\gamma}_j}}\Opw(\chi_{\kappa,|z|})u}\\
&\leq \parb{-8+2\tilde C\sqrt \epsilon+\delta}\inp{P'_2}_u +\delta^{-1}C_1\inp{P'_1}_u.
\end{split}
\end{align}

As for the analogous of  \eqref{eq:2_cont} we have, using the bound
$X_k^2\leq k$ and the first identity of \eqref{eq:GG}, and by arguing
as in \eqref{eq:second2}--\eqref{eq:second2b},
\begin{align}\label{eq:2_cont2}
  \begin{split}
  &2\sum_{i,j}\Re \inp{Q_{ij}u,X_k^{2(1-\epsilon')}\chi(r)\parbb{\tfrac {-2F^i}{
   r}(H-|z|)\gamma_j+\gamma_i\tfrac {-2F^j}{
   r}(H-|z|)}\Opw(\chi_{\kappa,|z|})u}\\
&\leq 2\Im z\,\inp{P_2}_u+\tilde C_1\tfrac {|(z-|z|)| ^2}{\Im z}\parb{\|X^{{-2}}X_k^{{2(1-\epsilon')}}u\|^2+k^2\inp{P'_1}_u}\\ & \quad \quad \quad + \delta\inp{P'_2}_{
    u}+\delta^{-1}C_2\parb{\|f^{-1/2}X^{3/2}v\|^2+\inp{P'_1}_u}.
\end{split}
\end{align}

As for the analogous of  \eqref{eq:3_cont} we obtain  by redistributing
components of $\gamma^2$, using the notation $u_{mn}=\Opw(\chi_{2\kappa,|z|})(4f/
   r)^{1/2}Q_{mn}$,
\begin{align}\label{eq:3_cont2}
  \begin{split}
  &2\sum_{i,j}\Re \inp{Q_{ij}u,X_k^{2(1-\epsilon')}\chi(r)\parbb{\tfrac {2F^i}{
   r}\gamma^2\gamma_j+\gamma_i\tfrac {2F^j}{
   r}\gamma^2}\Opw(\chi_{\kappa,|z|})u}\\
&\leq 2\sum_{m,n}\Re \inp{\tilde
  f \parb{\Opw(b)-1 +O(\epsilon^{3/4})}}_{u_{mn}}+ \tfrac \delta 2 \inp{P'_2}_{
    u}+C_{1,\delta}\|f^{-1/2}Xv\|^2\\
&\leq \parb{2\kappa (\sup |\tilde f|^2+1)+C_1\epsilon^{3/4}}\inp{4P'_2}_u + \delta\inp{P'_2}_{
    u}+C_{2,\delta}\|f^{-1/2}Xv\|^2\\
& 
=  \tilde C_3\epsilon^{3/4} \inp{P'_2}_u + \delta\inp{P'_2}_{
    u}+C_{2,\delta}\|f^{-1/2}Xv\|^2.
\end{split}
\end{align} Here we used
twice that $(rf)^{-1}=O\parb{r^{\mu/2-1}}$ (whence up to a compactly
supported term this function  is bounded by $\delta/C$), and we used  a uniform bound similar to
 \eqref{eq:calc3}, for example
\begin{align}\label{eq:calc3bb}
      \begin{split}
        [(4f/
   r)^{1/2}X_k^{2(1-\epsilon')}\chi(r)\gamma_m\gamma_n,\Opw(\chi_{2\kappa,|z|})]\Opw(\chi_{\kappa,|z|})&=X^{-2}BX^{-1}f^{1/2};\;
\|B\|\leq C.
      \end{split}
    \end{align}

As for the analogous of  \eqref{eq:4_cont} we have (using $l\geq 4$)
\begin{align}\label{eq:4_cont2}
  \begin{split}
  &2\sum_{i,j}\Re \inp{Q_{ij}u,X_k^{2(1-\epsilon')}\chi(r)\parbb{\i\tfrac {2f}{
   r^2}O_i(\epsilon^{0})\gamma_j+\gamma_i\i\tfrac {2f}{
   r^2}O_j(\epsilon^{0})}\Opw(\chi_{\kappa,|z|})u}\\
&=4\sum_{i,j}\Re \inp{Q_{ij}u,X_k^{2(1-\epsilon')}\chi(r)\parbb{\i\tfrac {2f}{
   r^2}O_i(\epsilon^{0})\gamma_j+\tfrac {f}{
   r^3}O_{ij}(\epsilon^{0})}\Opw(\chi_{\kappa,|z|})u}\\
&\leq \delta\inp{P'_2}_{
    u}+\delta^{-1}C_1 \parb{\|f^{-1/2}Xv\|^2+\inp{P'_1}_u}.
\end{split}
\end{align}

The analogue of \eqref{eq:S2_cont} is obvious.

The analogue of \eqref{eq:S3} reads with 
$u_{ij}:=\Opw(\chi_{2\kappa,|z|})((\phi-\phi_k)2f/r)^{1/2}Q_{ij}u$
\begin{align}\label{eq:S3bb}
  \begin{split}
  &2\sum_{i,j}\Re \inp{Q_{ij}u,\parb{\D X_k^{2(1-\epsilon')}}\chi(r)\gamma_i\gamma_j \Opw(\chi_{\kappa,|z|})u}\\
&\leq 4(1-\epsilon')\sum_{i,j}
\inp{\Opw(b)}_{u_{ij}}+C_{1}\|f^{-1/2}X v\|^2\\
&\leq 4(1-\epsilon')(1+2\kappa)\sum_{i,j} \|
u_{ij}\|^2+\delta\inp{P'_2}_u +C_2\|f^{-1/2}Xv\|^2\\
&\leq \parb{8(1-\epsilon')(1+2\epsilon^{3/4}) +\delta}\inp{P'_2}_u +C_3\|f^{-1/2}Xv\|^2.
\end{split}
\end{align} 
Here we used the bound \eqref{eq:calc3bb} trivially modified by an
insertion of a 
factor $(\phi-\phi_k)$.

Collecting bounds we get (similar to \eqref{eq:concl})
\begin{align}\label{eq:conclb}
  \begin{split}
  &\parb{8-2\tilde C \sqrt \epsilon-7\delta - \tilde C_3\epsilon^{3/4}
   -8(1-\epsilon')(1+2\epsilon^{3/4})}\inp{P'_2}_u\\
 &\leq C_{1}(\delta)\|f^{-1/2}X^{5/2}v\|^2+\tilde C_1\tfrac {|(z-|z|)| ^2}{\Im
  z}\parb{\|X^{{-2}}X_k^{{2(1-\epsilon')}}u\|^2+k^2\inp{P'_1}_u}+C_{2}(\delta)\inp{P'_1}_u.
 \end{split}
\end{align} 

Picking  $\delta=\sqrt \epsilon$ and $C_0= \tfrac14(\tilde
C+5)$ in  \eqref{eq:conclb} (as in Step I) the left hand side  bounds $\sqrt \epsilon\inp{P'_2}_u$ from
above, and whence \eqref{eq:Q1bnd2} follows.
\end{proof}
\section{Distorted Fourier transform}\label{sec:Distorted Fourier transform}
We prove the existence of the limit \eqref{eq:diag} for all
$\lambda\geq 0$ and all $v\in
  L^2_3$. For that we first compute the derivative $\tfrac{\d}{\d s}$ along the flow $\Phi (s,\cdot)$
\begin{equation}
  \label{eq:diag_deriv}
  \tfrac{\d}{\d s}\parb{\e^{-\i S_\epsilon} K^{1/2}_\epsilon m^{1/2}_\epsilon R(\lambda +\i 0)v}=\i\e^{-\i S_\epsilon} K^{-1/2}_\epsilon m^{1/2}_\epsilon \gamma_{\|}(\lambda)R(\lambda +\i 0)v.
\end{equation} It suffices to show that the right hand side is integrable as a
$\vG$-valued function. For the latter   purpose 
we use
 \eqref{eq:co_area}, the identity $S_\epsilon(\Phi (s,\cdot))=s$  and the Cauchy Schwarz
inequality to conclude  that in turn it suffices to find $\delta>0$ such that
\begin{equation}
  \label{eq:der_bnd}\|S_\epsilon^{1/2+\delta}K_\epsilon^{-1/2}\gamma_{\|}(\lambda)R(\lambda
  +\i 0)v\|\leq C<\infty.
  \end{equation} We plug in \eqref{eq:gamma_par}. Since
  \begin{equation}
    \label{eq:S_bnd}
    crf(r,\lambda)\leq S_\epsilon(x)\leq Crf(r,\lambda)
  \end{equation} the contribution from the first term $(H-V_2-\lambda)$ is
  in $\vH$ for any $\delta\leq 2$
  (then we have $(fr)^{1/2+\delta{}}K_\epsilon^{-1/2}\leq CX^3$ and by assumption
  $X^3v\in \vH$). For the contribution from the second  term
  $-\gamma^2$ we use \eqref{eq:Q_2_bound} with $i=j$. Again since $X^3v\in \vH$
  we only need to examine the weight $X^{2(1-\epsilon')}f^{1/2}r^{-1/2}$ to
  the left in \eqref{eq:Q_2_bound}, in particular we need to specify
  applicable $\delta$ and $\epsilon'$: More precisely  we need to
  specify these parameters such  that the function
  \begin{equation*}
    S_\epsilon^{1/2+\delta}K_\epsilon^{-1/2}X^{-2(1-\epsilon')}f^{-1/2}r^{1/2}\text{   is  bounded}.
  \end{equation*}  Using \eqref{eq:S_bnd} this
  property  will follow from boundedness of
  \begin{equation*}
    X^{2\epsilon' +\delta-1}f^{\delta-1},
  \end{equation*} which in turn for $\delta\leq 1$ will  follow from boundedness of 
\begin{equation*}
    X^{2\epsilon' +(\delta-1)(1-\mu/2)}.
  \end{equation*} The latter  boundedness  is achieved for any
  $\delta\in(0,1)$ (henceforth taken fixed)  and for all sufficient small $\epsilon,
  \epsilon'=C_0\sqrt \epsilon>0$. We have shown \eqref{eq:der_bnd} for all
$\lambda\geq 0$ and  all $v\in
  L^2_3$. Since the  bound \eqref{eq:der_bnd} is   uniform in $\lambda\in
  [0,\lambda_0]$ for any $\lambda_0>0$ and the function  $[0,\lambda_0]\ni
  \lambda\to \parb{\e^{-\i S_\epsilon} K^{1/2}_\epsilon
    m^{1/2}_\epsilon R(\lambda +\i 0)v}\parb{\Phi (s,\cdot)}\in \vG$
    is continuous for any (large) fixed  $s> 1$,  we conclude that also 
  the function 
  \begin{equation}\label{eq:contG}
    [0,\infty)\ni \lambda \to F^+(\lambda)v\in \vG\text{ is continuous}.
  \end{equation}

Clearly (by time-reversal invariance) we conclude the existence of
\eqref{eq:diag_-} for all $v\in
  L^2_3$ also. Similarly $F^-(\lambda)v$ is continuous in $\lambda\geq
  0$. 

There are other assertions in Subsection
\ref{subsec:Diagonalization}. As for the formula
\begin{align*}
  \|P_\c v\|^2=\lim_{\lambda_0\to \infty}\int _0^{\lambda_0}  \|F^+(\lambda)v\|_{\vG}^2\,\d \lambda,
\end{align*}  it suffices to show that for all $\lambda\geq 0$
\begin{align}\label{eq:fund_form}
  \pi^{-1}\inp{v, (\Im R(\lambda +\i
    0))v}= \|F^+(\lambda)v\|_{\vG}^2.
\end{align}  We first estimate (recall $u:=R(\lambda +\i 0)v$) 
 \begin{align}\label{eq:inf_arg}
   \begin{split}
  &|\lim_{s\to \infty}\Re \sum^d_{j=1}\int_{S^{d-1}}
\parb{\overline{(\gamma_j(\lambda)u)}u
  (\partial_jS_\epsilon)m_\epsilon}\parb{\Phi (s,\cdot)}\,\d \sigma|\\
&\leq\liminf_{s\to \infty}  \sum^d_{j=1}\int_{S^{d-1}}
|\parb{\overline{(\gamma_j(\lambda)u)}u
  (\partial_jS_\epsilon)m_\epsilon}\parb{\Phi (s,\cdot)}|\,\d \sigma.   
   \end{split}
\end{align} By \eqref{eq:co_area} and  the Cauchy Schwarz
inequality, for any big enough $s_0>0$
\begin{align*}
  &\int_{s_0}^\infty \d s \,s^{-1} \sum^d_{j=1}\int_{S^{d-1}}
|\parb{\overline{(\gamma_j(\lambda)u)}u
  (\partial_jS_\epsilon)m_\epsilon}\parb{\Phi (s,\cdot)}|\,\d \sigma\\
&=\sum^d_{j=1}
\int_{s_0}^\infty\d s \int_{S^{d-1}}
\big |m_\epsilon^{1/2}X^{1-\epsilon'}\parb {\tfrac{f_\lambda}{r}}^{1/2}\gamma_j(\lambda)u
  \big| \;\big |m_\epsilon^{1/2}X^{\epsilon'-1}\parb
  {\tfrac{f_\lambda}{r}}^{-1/2}(\partial_j\ln S_\epsilon)u\big |\,\d \sigma\\&\leq\sum^d_{j=1}\|X^{1-\epsilon'}\parb {\tfrac{f_\lambda}{r}}^{1/2}\chi(r)\gamma_j(\lambda)u\|\,\|X^{\epsilon'-1}\parb
  {\tfrac{f_\lambda}{r}}^{-1/2}\chi(r)(\partial_j\ln S_\epsilon)u\|.
\end{align*} Since
\begin{align*}
  |X^{\epsilon'-1}\parb{\tfrac{f_\lambda}{r}}^{-1/2}\chi(r)(\partial_j\ln S_\epsilon)|\leq CX^{\epsilon'+\mu/2-1}\parb{\tfrac{f_\lambda}{r}}^{1/2}\chi(r),
\end{align*}  cf.  \eqref{eq:S_bnd}, we conclude using
\eqref{eq:limitbound2j} and   \eqref{eq:Q_1_bound} that for all 
$\epsilon'=C_0\sqrt \epsilon  \in (0, 1-\mu/2)$  the latter
integral is finite. 
Whence for all small $\epsilon>0$ indeed the right hand
side of \eqref{eq:inf_arg} is zero.
We have shown \eqref{eq:fund_form}. 

Throughout the rest of the paper we abbreviate
$B=B(|x|)$  and $B^*=B(|x|)^*$. Note that due to \eqref{eq:limitbound2j} and \eqref{eq:fund_form}
\begin{equation}
  \label{eq:extbF}
  \forall \lambda \geq 0: \,F^+(\lambda)f^{1/2}\in \vB(B,\vG).
\end{equation} 

Next introduce
\begin{align*}
  F^+=\int _0^\infty \oplus F^+(\lambda)\,\d \lambda,
\end{align*} which due to \eqref{eq:fund_form} obeys
$(F^+)^*F^+=P_\c$. Notice that we here consider
$F^+\in\vB(\vH,\widetilde\vH)$. A short argument shows that for all
$v\in (H-\lambda)C_\c^\infty(\R^d)$ the function
$F^+(\lambda)v=0$. Whence $F^+H_\c\subset M_\lambda F^+$. We claim
that $F^+$ diagonalizes $H_\c$. This stronger statement is part of the
following  
\begin{proposition}
  \label{prop:dist-four-transf} The map $F^+:\Ran P_\c\to
  \widetilde\vH$ is a unitary diagonalizing transform, in particular  
 \begin{align}
  \label{eq:diagB}
  \Ran  F^+=\widetilde\vH\mand F^+H_\c=M_\lambda F^+.
\end{align}
\end{proposition}
\begin{proof}
  It suffices to show the first identity of \eqref{eq:diagB}, since
  then  indeed the restricted
  map   $F^+:\vH_\c(H)=\Ran P_\c\to \widetilde\vH$ is unitary and the
  second identity of \eqref{eq:diagB} holds.

\noindent {\bf Step I}
Let
 $\tau\in C^\infty(S^{d-1})$ and consider the function $\tilde u$ of
 \eqref{eq:tilde u}. We claim that
 \begin{equation}
   \label{eq:extbFcc}
   \tau= F^+(\lambda)(H-\lambda)\tilde u.
 \end{equation} Note that due to Lemma \ref{lemma:diagonalizationb}
 this is formally true, however since we dont know that $(H-\lambda)\tilde
 u\in L^2_3$ a continuity argument is required. This motivates the claim that for all $v\in B$
\begin{equation}
  \label{eq:extbF2}
  F^+(\lambda)f^{1/2}v=\vGlim_{S\to \infty}S^{-1}\int_0^S
  \pi^{-1/2}\parb{\e^{-\i S_\epsilon} K^{1/2}_\epsilon
    m^{1/2}_\epsilon R(\lambda +\i 0)f^{1/2}v}\parb{\Phi (s,\cdot)}\,\d s. 
\end{equation} Clearly this is consistent with \eqref{eq:extbF} if
$v\in f^{-1/2}L^2_3$. To show that indeed the  right hand side of
\eqref{eq:extbF2} makes sense for $v\in B$ we need to show the
Cauchy property. Approximating $C_\c(\R^d)\ni v_n\to v\in B$ it
suffices to show the bound 
\begin{align}
  \begin{split}
  \label{eq:1bndBa}
 \sup_{S>1}\|S^{-1}\int_0^S
  &\pi^{-1/2}\parb{\e^{-\i S_\epsilon} K^{1/2}_\epsilon
    m^{1/2}_\epsilon R(\lambda +\i 0)f^{1/2}(v-v_n)}\parb{\Phi
    (s,\cdot)}\,\d s\|_{\vG}\\&\leq C \|v-v_n\|_{B}.   
  \end{split} 
\end{align} We proceed a little more general and show for all $w\in B(|x|)^*$ 
\begin{align}
  \begin{split}
  \label{eq:1bndB}
 \sup_{S>1}\|S^{-1}\int_0^S
  \pi^{-1/2}\parb{\e^{-\i S_\epsilon} K^{1/2}_\epsilon
    m^{1/2}_\epsilon f^{-1/2}w}\parb{\Phi
    (s,\cdot)}\,\d s\|_{\vG}\leq C \|w\|_{B^*}.   
  \end{split} 
\end{align}
In fact given \eqref{eq:1bndB} the bound \eqref{eq:1bndBa}   follows using  
\eqref{eq:limitbound2j}, and whence the formula
\eqref{eq:extbF2} and then in turn \eqref{eq:extbFcc} are  justified.   

To show \eqref{eq:1bndB}  we first recall the  
 Besov space bound
\begin{subequations}
 \begin{align}\label{eq:besov1}
   \sup_{\rho>1}\rho^{-1}\int_{|x|\leq \rho}|w|^2\,\d x\leq C\|w\|^2_{B^*},
 \end{align} and its proof: Let $R_0=0$ and $R_j=2^{j-1} $ for
 $j\in\N$. Then
 \begin{align*}
   \rho^{-1}\int_{|x|\leq \rho}|w|^2\,\d x&\leq \sum_{j\leq
     J;\;R_{J-1}\leq \rho< R_J} \parb{\rho^{-1}R_j } R_j^{-1}\int_{R_{j-1}\leq|x|< R_j}|w|^2\,\d
   x\\&\leq\sum_{j\leq
     J;\;R_{J-1}\leq \rho< R_J} \parb{\rho^{-1}R_j } \|w\|^2_{B^*}\\&\leq 4\|w\|^2_{B^*}.
 \end{align*} Now for \eqref{eq:1bndB} we estimate using  the Cauchy Schwarz
 inequality, notation from Subsection
 \ref{subsec:Geometric properties},  \eqref{eq:co_area} and
 \eqref{eq:besov2} (stated below)
\begin{align*}
   &\|S^{-1}\int_0^S
  \pi^{-1/2}\parb{\e^{-\i S_\epsilon} K^{1/2}_\epsilon
    m^{1/2}_\epsilon f^{-1/2}w}\parb{\Phi
    (s,\cdot)}\,\d s\|_{\vG}\\&\leq C_1S^{-1}\int_0^S\|
  \parb{K^{1/2}_\epsilon
    m^{1/2}_\epsilon f^{-1/2}w}\parb{\Phi
    (s,\cdot)}\|_{\vG}\,\d s\\&\leq C_1S^{-1/2}\parb{\int_0^S\|
  \parb{K^{1/2}_\epsilon
    m^{1/2}_\epsilon f^{-1/2}w}\parb{\Phi
    (s,\cdot)}\|^2_{\vG}\,\d
  s}^{1/2}\\&=C_1S^{-1/2}\parb{\int_{\vB_\epsilon (S)}
  |K^{1/2}_\epsilon
   f^{-1/2}w|^2\,\d x}^{1/2}\\&\leq C_2S^{-1/2}\parb{\int_{\vB_\epsilon (S)}
  |f^{1/2}w|^2\,\d x}^{1/2}\\
&\leq C_3 \|w\|_{B^*}.
\end{align*} In the last step we used the following analogue of
\eqref{eq:besov1}:
\begin{align}\label{eq:besov2}
   \sup_{S>1}S^{-1}\int_{\vB_\epsilon (S)}|f^{1/2}w|^2\,\d x\leq C(\lambda)\|w\|^2_{B^*}.
 \end{align} Note that for $\lambda >0$ we can bound $S_\epsilon (x)\geq
 c|x|$ and $f^{1/2}(x)\leq C$ yielding \eqref{eq:besov2} 
  in this case  due to  \eqref{eq:besov1}. For $\lambda =0$ we can
 bound $S_\epsilon (x)\geq
 c|x|^{1-\mu/2}$ and $f^{1/2}(x)\leq C|x|^{-\mu/4}$ for $|x|\geq 1$
 yielding \eqref{eq:besov2} in that case also. This can be  seen by   
 arguing  as in the above proof of \eqref{eq:besov1}. Consequently we have
 \eqref{eq:besov2} for all $\lambda \geq0$, and \eqref{eq:1bndB} is proven.

For a later application let us note the following inverse of
\eqref{eq:besov2} (proved similarly):
\begin{align}\label{eq:besov3}
  \|w\|^2_{B^*}\leq  C(\lambda)\sup_{S>1}S^{-1}\int_{\vB_\epsilon (S)}|f^{1/2}w|^2\,\d x.
 \end{align} For an abstract version of \eqref{eq:besov2} and
 \eqref{eq:besov3} see \cite[Lemma 2.4]{Sk}.
\end{subequations}

\noindent {\bf Step II} We can mimic the proof of \cite[Theorem 1.1]{ACH}
using \eqref{eq:extbFcc}. Notice that we only need \eqref{eq:extbFcc}
for $\lambda >0$ (which is the analogue of  \cite[Theorem 3.3
iv)]{ACH}). Details are omitted.

\end{proof}

\begin{corollary}\label{cor:dist-four-transf}
  For all $\tau\in C^\infty(S^{d-1})\subset \vG$ the generalized
  eigenfunction $u^-=u^-(\lambda)$, $\lambda\geq 0$,  defined  by \eqref{eq:tilde u} and
  \eqref{eq:gen_eig} is also given by
  \begin{align}
    \label{eq:form_gene_func}
    u^-(\lambda)=2\pi\i F^+(\lambda)^*\tau.
  \end{align}
\end{corollary}
\begin{proof}
  By Lemma \ref{lemma:diagonalizationb}, \eqref{eq:fund_form} and \eqref{eq:extbFcc}
  \begin{align*}
   u^-&=\parb{R(\lambda+\i 0)-R(\lambda-\i 0)} (H-\lambda)\tilde u\\&=2\pi\i F^+(\lambda)^*F^+(\lambda)(H-\lambda)\tilde u=2\pi\i F^+(\lambda)^*\tau.
  \end{align*}
\end{proof}

\begin{defn}
  \label{defn:dist-four-transf} For any $\lambda\geq 0$ we define the
  scattering matrix $S(\lambda)\in \vB (\vG)$ by the identity
  \begin{align}
    \label{eq:1S} F^+(\lambda)v=S(\lambda)F^-(\lambda)v;\;v\in  f_\lambda^{1/2}B(|x|).
    \end{align}
\end{defn} 
\begin{proposition}\label{prop:dist-four-transf2} The operator
  $S(\lambda)$ is a well-defined unitary operator on $\vG$. It is
  strongly continuous as a function of $\lambda\geq 0$. In particular
  the  scattering matrix at zero energy $S(0)$ is uniquely
  determined by the diagonalizing transforms $F^{\pm}$.
  \end{proposition}
  \begin{proof}
    We apply \eqref{eq:fund_form}, \eqref{eq:extbF}, \eqref{eq:extbFcc} and their
    analogues for change of superscript $+\to -$. This yields the
    well-definedness and the unitarity. For all $v\in L^2_3$ the
    functions $\{F^{\pm}(\lambda)v|\lambda\geq 0\}\in \widetilde \vH$
    are continuous in $\lambda$, cf. \eqref{eq:contG}. Since moreover  
    $F^{-}(\lambda)L^2_3$ is dense in $\vG$ for any fixed $\lambda$ the continuity property follows by a
    density argument.
  \end{proof}

\subsection{Asymptotics of 
generalized eigenfunctions} \label{subsec:Asymptotics of 
generalized eigenfunctions}
We complete this section by a discussion of the asymptotics of the
generalized eigenfunctions
\begin{equation}\label{eq:GENERA_FUNC}
  u^-_\tau(\cdot,\lambda):=2\pi\i F^+(\lambda)^*\tau;\;\tau\in \vG.
\end{equation} Notice that Corollary  \ref{cor:dist-four-transf}
provides a representation for $\tau\in  C^\infty(S^{d-1})$.

Let $B_0^*\subset B^*$ be the closure of $C_{\c}(\R^d)$ in $B^*$. For
all $\lambda\geq0$ and $\tau\in \vG$ the function
\begin{equation*}
w(x)=\parb{K_\epsilon^{-1/2}m_\epsilon^{-1/2} f^{1/2}}(x)\tau(\sigma);\;x=\Phi
    (t,\sigma),  
\end{equation*} belongs to $ B^*$ with 
\begin{equation}
  \label{eq:fund_bound}
  \|w\|_{B^*}\leq
C\|\tau\|_\vG.
\end{equation}
  This is due to \eqref{eq:co_area} and   \eqref{eq:besov3}.

 Next,   using  \eqref{eq:extbFcc}
and  \eqref{eq:1S}  we
  decompose  
for all $\tau\in  C^\infty(S^{d-1})$
\begin{align*}
  w^-_{\tau}(x):&=\pi^{1/2}\parb{K_\epsilon^{-1/2}m_\epsilon^{-1/2} f^{1/2}}(x)\parb{S(\lambda)^{-1}\tau}(\sigma)\\&=\pi^{1/2}\parb{K_\epsilon^{-1/2}m_\epsilon^{-1/2} f^{1/2}}(x)\parb{F^-(\lambda) (H-\lambda)\tilde u}(\sigma)\\
&=w^-_1(x)+w^-_2(x);\\&\quad w^-_1:=\e^{\i S_\epsilon}f^{1/2}R(\lambda-\i 0)(H-\lambda)\tilde u.
\end{align*} While   $w^-_{\tau}, w^-_1\in B^*$ we have the  stronger
assertion for  the second term, 
\begin{equation}
  \label{eq:rem_est}
  w^-_2\in
B_0^*.
\end{equation} To prove  \eqref{eq:rem_est}  we introduce  the quantity
\begin{align*}
  w_n=\pi^{1/2}\parb{K_\epsilon^{-1/2}m_\epsilon^{-1/2} f^{1/2}}(x)\parb{F^-(\lambda) \parb{(H-\lambda)\tilde u-f^{1/2}v_n}}(\sigma),
\end{align*} where $C_{\c}(\R^{d})\ni v_n\to v:=f^{-1/2}(H-\lambda)\tilde u\in
B$. We have $\|w_n\|_{B^*}\leq C\|v_n-v\|_B$,
cf. \eqref{eq:fund_bound},  showing that
$\|w_n\|_{B^*}\to 0$ for $n\to \infty$. Similarly
\begin{align*}
  \e^{\i S_\epsilon}f^{1/2}R(\lambda-\i 0)f^{1/2}v_n\to
  w^-_1\text{ in }B^*.
\end{align*}
We are lead to consider the quantity (for fixed $n\in \N$)
\begin{align*}
  \widetilde w_n(x)=\pi^{1/2}\parb{K_\epsilon^{-1/2}m_\epsilon^{-1/2} f^{1/2}}(x)\parb{&\parb{F^-(\lambda) f^{1/2}v_n}(\sigma)\\&-\pi^{-1/2}\parb{\e^{\i S_\epsilon}K_\epsilon^{1/2}m_\epsilon^{1/2} R(\lambda-\i 0)f^{1/2}v_n}(x)},
\end{align*} It follows  from   \eqref{eq:co_area}, \eqref{eq:diag_deriv},
\eqref{eq:der_bnd} and   \eqref{eq:besov3}  by yet another
approximation that $\widetilde w_n\in
B_0^*$. Note that indeed 
\begin{align*}
 \|\widetilde w_n-1_{\vB_\epsilon(S)}\widetilde w_n\|_{B^*}\to
0 \text{ for }S\to \infty,
\end{align*} while obviously we have $1_{\vB_\epsilon(S)}\widetilde w_n\in
B_0^*$ for all $S>1$.
Whence \eqref{eq:rem_est} is proven.

Now, combining    
  \eqref{eq:gen_eig}  and \eqref{eq:rem_est} we conclude that for all $\tau\in  C^\infty(S^{d-1})$
  \begin{align*}
    u^-_\tau(\cdot,\lambda)-\parb{\tilde u-\e^{-\i S_\epsilon}f^{-1/2}w^-_\tau}\in
f^{-1/2}B_0^*.
\end{align*} This formula extends to $\vG$ and implies an injectivety property.
  \begin{corollary}\label{cor:dist-four-transfasyP}
    Let  $\lambda\geq0$ and  $\tau\in \vG$ be given. Introducing  the
    following function of  $x=\Phi
    (t,\sigma)$,
    \begin{align*}
      u^-_{0,\tau}(x,\lambda)=\pi^{1/2}\parb{K_\epsilon^{-1/2}m_\epsilon^{-1/2}
    }(x)\parb{\e^{\i S_\epsilon(x)}\tau(\sigma)-\e^{-\i
        S_\epsilon(x)}\parb{S(\lambda)^{-1}\tau}(\sigma)},
    \end{align*}
we have
\begin{subequations}
\begin{align}\label{eq:aEigenf}
    u^-_\tau(\cdot,\lambda)-u^-_{0,\tau}(\cdot,\lambda)\in
f_\lambda^{-1/2}B_0^*.
  \end{align} In particular
\begin{align}\label{eq:aEigenf2}
    2\pi \|\tau\|_{\vG}^2=\lim_{S\to \infty}S^{-1}\int_{\vB_\epsilon (S)}
  |K_\epsilon^{1/2}u^-_{\tau}(x,\lambda)|^2\,\d x.
  \end{align}
  \end{subequations}
\end{corollary}
\begin{proof}
  Since we know \eqref{eq:aEigenf} for $\tau\in  C^\infty(S^{d-1})$
  the statement for $\tau\in \vG$ 
  follows  by approximation, cf. \eqref{eq:fund_bound}. As for \eqref{eq:aEigenf2} we can replace
  $u^-_{\tau}$ to the right by $u^-_{0,\tau}$ due to
  \eqref{eq:aEigenf}. Then we compute using \eqref{eq:co_area} and the
  unitarity of the scattering matrix. Note
  that cross terms do not contribute to the
  limit due to oscillatory behaviour.
\end{proof}

\section{Characterization of generalized eigenfunctions}\label{sec:Characterization of generalized eigenfunction}
We introduce the following class of generalized eigenfunctions.
\begin{defn}\label{def:char-gener-eigenf}
  For  $\lambda\geq0$ let 
  \begin{equation*}
    \vE_\lambda=\{u\in f_\lambda^{-1/2}B^*| (H-\lambda)u=0\}.
  \end{equation*}
\end{defn}
Notice that  it follows from \eqref{eq:extbF} and the definition
\eqref{eq:GENERA_FUNC} that for all $\tau\in \vG$ 
\begin{equation*}
  u^-_\tau(\cdot,\lambda)\in\vE_\lambda.
\end{equation*} In fact it follows from \eqref{eq:extbF},
\eqref{eq:besov2} and 
\eqref{eq:aEigenf2} that the map 
\begin{equation*}
  \vG \ni\tau \to u^-_\tau(\cdot,\lambda)\in \vE_\lambda
\end{equation*} is a bi-continuous linear isomorphism onto its range.
The latter is identified as
\begin{proposition}\label{prop:char-gener-eigenf-1}  For all $\lambda\geq 0$ the set 
  \begin{align}
    \label{eq:spaces}
    \vE_\lambda=\{u^-_\tau(\cdot,\lambda)|\tau\in\vG\}.
  \end{align}
\end{proposition}
\begin{proof} Let $u_\lambda\in \vE_\lambda$ be arbitrarily given. We  need
  to show that it must have the form
  $u_\lambda=u^-_\tau(\cdot,\lambda)$ for some $\tau\in\vG$. For that
  we partly mimic
  \cite[Section 8]{DS3}. In particular,   with reference to symbols
  \eqref{eq:ab_lambda} and the corresponding localization symbols as
  appearing in Proposition \ref{prop:resol_bounds} let us introduce
  \begin{align}\label{eq:dec_unity}
  \chi^\pm=\chi_-(a_\lambda)\tilde\chi_\pm(b_\lambda)+\tfrac12\chi_+(a_\lambda).  
  \end{align} We consider in the following these functions as fixed
  and consider $\epsilon$-small perturbations $W_\epsilon\in \vW$ with
  $\epsilon>0$ small exactly as in  Proposition \ref{prop:resol_bounds}. Note the properties
  \begin{align}
    \label{eq:energ_forb}
   \Opw \parb{a_\lambda\chi_+(a_\lambda)}u_\lambda,\,\Opw \parb{\chi_+(a_\lambda)}u_\lambda\in
f_\lambda^{-1/2}B_0^*,
  \end{align} cf. \cite[Lemma 3.1]{Sk}, in fact these functions are in
  any weighted 
  $L^2$-space $L^2_m$. Whence the quantization of the second term of \eqref{eq:dec_unity}
  contributes by a small term when applied to $u_\lambda$. The quantization of the first term  localizes to an outgoing
  (incoming) region of phase space. A priori we only have
\begin{align*}
    \Opw \parb{\chi^\pm}u_\lambda,\,\Opw \parb{\chi_-(a_\lambda)\tilde\chi_\pm(b_\lambda)}u_\lambda\in
f_\lambda^{-1/2}B^*.
  \end{align*}
\noindent {\bf Step I} We construct a candidate $\tau$. Pick a
non-negative $g\in C_\c^\infty(\R_+)$  with $\int_0^\infty g(t)\d t=1$,
and let $G_n(s)=1-\int_0^{s/n} g(t)\d t$, $n\in\N$. Define
\begin{align*}
  \tau_n=F^+(\lambda)G_n(S_\epsilon)(H-\lambda)\Opw \parb{\chi^+}u_\lambda;\,n\in\N.
\end{align*} We note that this family $\{\tau_n\}$ is a bounded subset
of $\vG$. In fact  we have 
\begin{align*}
  \tau_n&=\i F^+(\lambda)\i[H,G_n(S_\epsilon)]\Opw \parb{\chi^+}u_\lambda\\
&=-\i F^+(\lambda)\parb{\Re(p\cdot \nabla S_\epsilon)\tfrac 2n
  g(S_\epsilon/n)+\i|\nabla S_\epsilon|^2 n^{-2}g'(S_\epsilon/n)}\Opw \parb{\chi^+}u_\lambda\\
&=\tau_n^1+\tau_n^2.
\end{align*} For any $\tilde \tau\in \vG$ we have
\begin{subequations}
\begin{align}
  \label{eq:tilde_part}
  \|f_\lambda^{1/2} u^-_{\tilde \tau}(\cdot, \lambda)\|_{B^*}&\leq C\|\tilde
  \tau\|_{\vG},\\
\|f_\lambda^{-3/2}\chi(|x|>2)\Re(p\cdot \nabla
  S_\epsilon)u^-_{\tilde \tau}(\cdot, \lambda)\|_{B^*}&\leq C\|\tilde \tau\|_{\vG},\label{eq:tilde_part2}
\end{align} cf. \eqref{eq:energ_forb}. We aim at showing the uniform bounds 
\end{subequations}
\begin{align} 
  \label{eq:unif}
  |\inp{\tilde \tau, \tau_n^j}|\leq C \|\tilde \tau\|_{\vG};\,j=1,2,
\end{align} which suffices for the boundedness. 

For $j=2$ we write
\begin{align*}
  -2\pi\i  \inp{\tilde \tau, \tau_n^2}&=\inp{f_\lambda^{1/2} u^-_{\tilde
      \tau}(\cdot, \lambda), h^2_n w};\\&h^2_n=f_\lambda^{-1/2} |\nabla S_\epsilon|^2 n^{-2}g'(S_\epsilon/n)f_\lambda^{-1/2},\,w=f_\lambda^{1/2} \Opw \parb{\chi^+}u_\lambda.
\end{align*} For $\lambda>0$ we have $|h^2_n(x)|\leq C_\lambda
\inp{x}^{-2}$ while for  $\lambda=0$ there is the bound  $|h^2_n(x)|\leq C
\inp{x}^{-2+\mu/2}$. Since $B^*$ is continuously imbedded in
$L^2_{-\delta}$ for any $\delta>1/2$ we conclude the bound
\eqref{eq:unif} for $j=2$ using
\eqref{eq:tilde_part}.

Decomposing similarly for $j=1$, 
\begin{align*}
  \inp{\tilde \tau, \tau_n^1}&=\inp{\tilde w, h^1_n w};\\&h^1_n=f_\lambda^{3/2} n^{-1}g(S_\epsilon/n)f_\lambda^{-1/2},\,w=f_\lambda^{1/2} \Opw \parb{\chi^+}u_\lambda,
\end{align*} and using \eqref{eq:tilde_part2} we
need to show the bound
\begin{align*}
  |\inp{\tilde w, h^1_n w}|\leq C
  \|\tilde w\|_{B^*}\,\| w\|_{B^*}.
\end{align*} For that it suffices for any $\lambda\geq 0$ to find $C>1$ and a bounded  interval $I$
 such that for all  $n\in\N$ there exists $R\geq 1$ such that $|h^1_n(x)|\leq C
R^{-1}1_I(|x|/R)$ for all $x\in\R^d$. We recall the bounds $crf_\lambda \leq
S_\epsilon\leq Crf_\lambda$. In particular for $\lambda>0$ the assertion is
immediate with $R=n$. For $\lambda=0$ we choose $R=n^{1/(1-\mu/2)}$
and obtain the same conclusion. So indeed  \eqref{eq:unif} holds, and
the sequence $\{\tau_n\}\subset\vG$ is  bounded.

Take $\tau\in\vG$ as the weak limit of some subsequence of
$\{\tau_n\}$, cf. \cite[Theorem 1 p. 126]{Yo}. Upon changing the notation we can assume that
\begin{align}
  \label{eq:tau}
  \tau=\wvGlim_{n\to \infty} F^+(\lambda)G_n(S_\epsilon)[H,\Opw \parb{\chi^+}]u_\lambda.
\end{align}

\noindent {\bf Step II} We show that this $\tau$ works. We compute
using \eqref{eq:fund_form} in the third  step, and Propositions
\ref{prop:eikonal-equation} and 
\ref{prop:resol_bounds} in the last step,  and taking  $m=-3$,
\begin{align*}
  f^{1/2}u^-_\tau(\cdot,\lambda)=&2\pi\i f^{1/2}F^+(\lambda)^*\tau\\
=&2\pi\i \wBstarlim_{n\to
  \infty} f^{1/2}F^+(\lambda)^*F^+(\lambda)G_n(S_\epsilon)[H,\Opw \parb{\chi^+}]u_\lambda\\
=& \wBstarlim_{n\to \infty} f^{1/2}\parb{  R(\lambda-\i 0)-R(\lambda+\i
  0)}G_n(S_\epsilon)[\Opw \parb{\chi^+},H-\lambda]u_\lambda\\
=& \wLstarlim_{n\to \infty} f^{1/2}R(\lambda-\i
0)G_n(S_\epsilon)[H-\lambda,\Opw \parb{\chi^-}]u_\lambda\\
&+ \wLstarlim_{n\to \infty} f^{1/2}R(\lambda+\i
0)G_n(S_\epsilon)[H-\lambda,\Opw \parb{\chi^+}]u_\lambda\\
=&w^-+w^+;\,w^\mp=f^{1/2}R(\lambda\mp\i
0)(H-\lambda)\Opw \parb{\chi^\mp}u_\lambda.
\end{align*}

Using again Proposition
\ref{prop:resol_bounds} we compute
\begin{align*}
  w^\mp&=\wLstarlim_{\epsilon\searrow 0}f^{1/2}R(\lambda\mp\i
  \epsilon)(H-\lambda)\Opw \parb{\chi^\mp}u_\lambda\\&=f^{1/2}\Opw \parb{\chi^\mp}u_\lambda\mp \wLstarlim_{\epsilon\searrow 0}\i\epsilon
f^{1/2}R(\lambda\mp\i 
\epsilon)\Opw \parb{\chi^\mp}u_\lambda\\
&=f^{1/2}\Opw \parb{\chi^\mp}u_\lambda.
\end{align*} Whence
\begin{align*}
  u^-_\tau(\cdot,\lambda)=f^{-1/2}\parb{w^-+w^+}=\parbb{\Opw \parb{\chi^-}+\Opw \parb{\chi^+}}u_\lambda=u_\lambda.
\end{align*}
\end{proof}

We summarize
\begin{thm}
  \label{thm:char-gener-eigenf-1} Suppose Condition \ref{cond:rad}
  (and 
  \eqref{eq:supp_rad}). There exist  $\epsilon_0>0$ and
  $l=l(\mu,d)\in\N$  such
  that for  all $\epsilon$-small perturbations $W_\epsilon$
    with $\epsilon\in
  (0,\epsilon_0]$ (assuming also
   \eqref{eq:supp_pert}) the
  following statements hold for all  $\lambda\geq0$: 

For
  all 
  $\tau\in \vG$ there exist   unique $\tilde \tau\in \vG$ and 
    $u_\lambda\in \vE_\lambda$ such that (with $x=\Phi
    (t,\sigma)$)
\begin{subequations}
    \begin{align}\label{eq:gen1}
      u_\lambda(x)-\pi^{1/2}\parb{K_\epsilon^{-1/2}m_\epsilon^{-1/2}
    }(x)\parb{\e^{\i S_\epsilon(x)}\tau(\sigma)-\e^{-\i
        S_\epsilon(x)}\tilde\tau(\sigma)}\in
f_\lambda^{-1/2}B_0^*.
    \end{align}

Moreover for all $u_\lambda\in \vE_\lambda$ there exist  unique $\tau,
\tilde \tau\in \vG$ such that \eqref{eq:gen1} holds.
In particular the map $
  \vG \ni\tau \to u_\lambda\in \vE_\lambda
$ is a linear isomorphism. It is  bi-continuous, in fact 
   \begin{align}\label{eq:aEigenf2w}
    2\pi \|\tau\|_{\vG}^2=\lim_{S\to \infty}S^{-1}\int_{\vB_\epsilon (S)}
  |K_\epsilon^{1/2}u_{\lambda}|^2\,\d x.
  \end{align} 

There are formulas 
\begin{align}\label{eq:aEigenfw}
    u_\lambda=u^-_\tau(\cdot,\lambda)=2\pi\i F^+(\lambda)^*\tau\mand \tilde \tau=S(\lambda)^{-1}\tau.
  \end{align} In particular the wave matrix $F^+(\lambda)^*:\vG\to
  \vE_\lambda$ is a bi-continuous linear isomorphism. The maps $F^+(\lambda)f_\lambda^{1/2}:B\to
  \vG$ and $\delta(\lambda)=\pi^{-1}\Im \parb{R(\lambda+\i 0)}:f_\lambda^{1/2}B\to
  \vE_\lambda$ are  onto.
\end{subequations}
\end{thm}
\begin{proof} The uniqueness of $\tilde \tau$ and $u_\lambda$ in
  \eqref{eq:gen1} follows from the proof of Lemma~\ref{lemma:diagonalizationb}, and the
  existence  part
   (in agreement with \eqref{eq:aEigenfw}) follows from  \eqref{eq:aEigenf}. The mapping
  properties mentioned in the last sentence of the theorem are
  consequences  of
  Banach's closed range theorem   \cite[Theorem  p. 205]{Yo} and  previous
  statements. The remaining parts of the latter 
  are consequences of \eqref{eq:aEigenf2}  and Proposition 
  \ref{prop:char-gener-eigenf-1}.
\end{proof}

\subsection{Concluding  remarks}\label{sec:Concluding comments}
With some more effort  one should be able to show that the operator
$F^+(\lambda)^*$ has a somewhat regular kernel, formally given by
$(F^+(\lambda)^*\delta_\sigma)(x)$. More precisely one should have
\begin{align*}
 (F^+(\lambda)^*\tau)(x)
 =\int_{\vG}\phi^+(x,\sigma,\lambda)\tau(\sigma)\,\d \sigma,
\end{align*} where the plane wave type eigenfunction $\phi^+$ has 
a degree of regularity. In particular it should be continuous in all
variables for $ x\notin \supp V_2$ provided the perturbation
$W_\epsilon\in \vW$ is $\epsilon$-small with $\epsilon>0$ taken small
enough. More regularity in $\sigma\in \vG$ would require $\epsilon$
taken smaller. These assertions depend on possible generalizations of
Proposition \ref{Prop:radi-cond-bounds}, cf. \cite{HS2}. We shall not
elaborate further on this issue. Note also
that smoothness in the angular
variable  of analogous plane wave type eigenfunctions was indeed
obtained in \cite{DS3}. 

Another remark concerns the relationship between the scattering theory
developed here and \cite{DS3} in case of overlapping conditions
(which means under the conditions of \cite{DS3}). In the case of a
spherically  symmetric potential we have for all $\lambda\geq 0$ that 
\begin{align*}
\sigma=\eta(\sigma):=\lim _{s\to \infty}\Phi(s,\sigma)/|\Phi(s,\sigma)|,  
\end{align*} and the two involved solutions to the eikonal equation
are identical up to a trivial explicit  term. In
particular the two families of $S$-matrices are explicitly  connected
as follows: 
 For all $\lambda\geq 0$ the operator $S(\lambda)$ of this paper
and the scattering matrix of \cite{DS3}, say $S_{\rm DS}(\lambda)$,
are up to an explicit phase factor
related  as $S(\lambda)=S_{\rm DS}(\lambda)R$ where 
$(R\tau)(\omega)=\tau(-\omega)$, cf.  the discussion at the 
beginning of Section \ref{sec:introduction}.

More generally under  the conditions of \cite{DS3} the
asymptotic normalized velocity $\eta(\sigma)$ exists for all
$\lambda\geq 0$ and as a map it is
a diffeomorphism on $S^{d-1}$. This property, Theorem
\ref{thm:char-gener-eigenf-1} and \cite[Theorem 8.2]{DS3} yields the connection formula
\begin{align}\label{eq:Stwo}
  S_{\rm DS}(\lambda)^{-1}=R\e^{-\i \phi(\cdot, \lambda)}D_\eta S(\lambda)^{-1}D_{\eta^{-1}}\e^{-\i \phi(\cdot, \lambda)},
\end{align} where $\phi(\omega, \lambda)$ is  real and for any
diffeomorphism $\psi$ on $S^{d-1}$ the operator $D_\psi$ is the
unitary map on $L^2(S^{d-1})$ implemented by the classical map 
$S^{d-1}\ni\omega \to \psi(\omega)\in S^{d-1}$, viz. $(D_\psi\tau)(\omega)=J^{1/2}(\omega)\tau(\psi^{-1}(\omega))$.

Although we shall not elaborate the formula \eqref{eq:Stwo} suggests  a
criterion  for regularity at zero energy of a family of (inverse) scattering
matrices under the conditions of Section \ref{sec:Smaller class}: It suffices that the
families of diffeomorphisms $\eta=\eta_\lambda$ and
$\eta_\lambda^{-1}$ on $S^{d-1}$  as well as the family of phases
$\phi(\cdot, \lambda)$ are
regular at zero energy. Indeed in  this case the right hand side of
\eqref{eq:Stwo} has a limit as $\lambda\to 0$ due to Proposition \ref{prop:dist-four-transf2}. This criterion  is of
course not applicable for the example in Subsection \ref{sec:Conditions}.
Note that  under the conditions of Section
\ref{sec:Smaller class}  the form of the right hand
side of \eqref{eq:Stwo} makes sense for positive energies and
the expression coincides 
with the (inverse) scattering matrix discussed in the beginning of
Section \ref{sec:introduction}.

\end{document}